\documentclass[11pt]{article}

    \usepackage[T1]{fontenc}
    \usepackage{lmodern}
    \usepackage{amsmath,amssymb,amsthm,mathtools}
    \usepackage{microtype}
    \usepackage[margin=1in]{geometry}
    \usepackage{enumitem}
    \usepackage{aliascnt}
    \usepackage{algorithm}
\usepackage{algpseudocode}
    \usepackage[hidelinks]{hyperref}
    \hypersetup{pdftitle={The Exact Online Threshold for the Asymmetric Binary Perceptron}}
    \usepackage[nameinlink,noabbrev]{cleveref}

    \allowdisplaybreaks

    \newtheorem{theorem}{Theorem}[section]

    \newaliascnt{proposition}{theorem}
    \newtheorem{proposition}[proposition]{Proposition}
    \aliascntresetthe{proposition}

    \newaliascnt{lemma}{theorem}
    \newtheorem{lemma}[lemma]{Lemma}
    \aliascntresetthe{lemma}

    \newaliascnt{corollary}{theorem}
    \newtheorem{corollary}[corollary]{Corollary}
    \aliascntresetthe{corollary}

    \theoremstyle{definition}
    \newaliascnt{definition}{theorem}
    \newtheorem{definition}[definition]{Definition}
    \aliascntresetthe{definition}

    \theoremstyle{remark}
    \newaliascnt{remark}{theorem}
    \newtheorem{remark}[remark]{Remark}
    \aliascntresetthe{remark}

    \crefname{theorem}{Theorem}{Theorems}
    \Crefname{theorem}{Theorem}{Theorems}
    \crefname{proposition}{Proposition}{Propositions}
    \Crefname{proposition}{Proposition}{Propositions}
    \crefname{lemma}{Lemma}{Lemmas}
    \Crefname{lemma}{Lemma}{Lemmas}
    \crefname{corollary}{Corollary}{Corollaries}
    \Crefname{corollary}{Corollary}{Corollaries}
    \crefname{definition}{Definition}{Definitions}
    \Crefname{definition}{Definition}{Definitions}
    \crefname{remark}{Remark}{Remarks}
    \Crefname{remark}{Remark}{Remarks}
    \crefname{section}{Section}{Sections}
    \Crefname{section}{Section}{Sections}
    \crefname{appendix}{Appendix}{Appendices}
    \Crefname{appendix}{Appendix}{Appendices}
    \crefname{algorithm}{Algorithm}{Algorithms}
    \Crefname{algorithm}{Algorithm}{Algorithms}

    \newcommand{\R}{\mathbb{R}}
    \newcommand{\N}{\mathbb{N}}
    \newcommand{\E}{\mathbb{E}}
    \newcommand{\Pp}{\mathbb{P}}
    \newcommand{\1}{\mathbf{1}}
    \newcommand{\cF}{\mathcal{F}}
    \newcommand{\cG}{\mathcal{G}}
    \newcommand{\cH}{\mathcal{H}}
    \newcommand{\cL}{\mathcal{L}}
    \newcommand{\norm}[1]{\lVert #1\rVert}
    \newcommand{\abs}[1]{\lvert #1\rvert}
    \newcommand{\esssup}{\operatorname*{ess\,sup}}
    \newcommand{\dist}{\operatorname{dist}}
    \newcommand{\aon}{\alpha_{\mathrm{on}}}
    \newcommand{\Estar}{E_\star}
    \newcommand{\cgauss}{c_{\mathrm G}}
    \newcommand{\ca}{c_{\alpha}}
    \newcommand{\ph}{\varphi}
    \newcommand{\Ph}{\Phi}

    \title{The Exact Online Threshold for the Asymmetric Binary Perceptron}
    \author{%
  Sunghyeon Jo\\
  {\small Georgia Institute of Technology and QED Audit}\\
  {\small\texttt{sjo65@gatech.edu}}
  \and
  Taekyun Lee\\
  {\small The University of Texas at Austin}\\
  {\small\texttt{taekyun@utexas.edu}}
}
    \date{August 2026}

\begin{document}
    \maketitle
 \begin{abstract}

  Let $G\in\R^{M\times N}$ have independent standard Gaussian entries.  For a
  fixed margin $\kappa\in\R$, the asymmetric binary perceptron asks for
  $\sigma\in\{\pm1\}^N$ such that $G\sigma/\sqrt N\ge\kappa\1_M$.  We study
  the online version of this problem, in which the columns of $G$ arrive
  sequentially and each sign must be chosen irrevocably before future columns
  are revealed.

  We determine the exact threshold $\aon(\kappa)$ for every fixed $\kappa$:
  for $M/N\to\alpha$ with $\alpha<\aon(\kappa)$, there is a deterministic
  online algorithm, using $O(MN)$ arithmetic operations and polynomial
  bit complexity, that succeeds with high probability, while for $\alpha>\aon(\kappa)$, no online algorithm succeeds
  with high probability.  The threshold is characterized by a one-dimensional
  stochastic control problem for Brownian motion.  The main difficulty is to upgrade a single-coordinate
  Brownian limit to simultaneous feasibility of all $M=\Theta(N)$ constraints,
  which we do with half-line monotonicity and a short final correction block.

  At zero margin, we give a computer-assisted proof that
  $0.32747<\aon(0)<0.36664$.  
  In particular, every density below $0.32747$ is achievable online by such
  an algorithm, more than tripling the best density previously proved
  attainable by any polynomial-time algorithm, online or offline
  ($\alpha\le0.1$, Li--Schramm--Zhou~\cite{LiSchrammZhou25}).
  As
  $\kappa\to+\infty$, the online threshold agrees to first order with the
  offline storage capacity.  As $\kappa\to-\infty$, it has the same
  asymptotic scale as the best known offline polynomial-time guarantee, while
  the storage capacity is larger by a factor of order $\kappa^2$.
  \end{abstract}

\section{Introduction}
\label{sec:introduction}

The asymmetric binary perceptron (ABP), also called the asymmetric Ising perceptron, is a dense random constraint satisfaction problem in which one seeks a vertex of the hypercube lying in the intersection of random half-spaces.  Given $G\in\R^{M\times N}$ with independent standard Gaussian entries and a margin $\kappa\in\R$, it asks for $\sigma\in\{\pm1\}^N$ with $G\sigma/\sqrt N\ge\kappa\1_M$.  When $M/N\to\alpha$, the classical feasibility question asks for the supremal constraint density at which such a $\sigma$ exists with high probability.  This critical density is the storage capacity (or satisfiability threshold) of the ABP~\cite{KrauthMezard89,Xu21,NakajimaSun23}.  A second line of work asks for the largest constraint density achievable by efficient algorithms~\cite{KimRoche98,LiSchrammZhou25}.

In this paper, we impose an additional information constraint.  Write $G=(g_1,\ldots,g_N)$ by columns.  At time $k$, an algorithm observes $g_1,\ldots,g_k$ and must choose $\sigma_k\in\{\pm1\}$ before $g_{k+1}$ is revealed.  The choice is irrevocable.  We ask for the supremal constraint density at which an algorithm with this information pattern can satisfy all $M$ inequalities with probability tending to one.  Denote this threshold by $\aon(\kappa)$.

This online model is closely related to online stochastic vector balancing: a column is an arriving Gaussian vector and its sign must be chosen immediately.  The objective here, however, is feasibility rather than minimization of a symmetric discrepancy.  The one-column geometry already suggests the limiting control problem. If $Z\sim N(0,I_M/N)$ is the next normalized Gaussian column, then, conditional on the information revealed so far, every rule for choosing $\varepsilon\in\{\pm1\}$ after observing $Z$ induces a mean drift $b$ satisfying
\[
    \norm{b}_2\le\sqrt{\frac{2}{\pi N}}.
\]
Conversely, every predictable $b$ in this Euclidean ball can be realized by a sign choice while retaining a fresh Gaussian innovation.  When $M/N\to\alpha$, the corresponding squared $L^2$ budget per coordinate is
\[
    \frac{2}{\pi\alpha}.
\]
This leads to the following one-dimensional Brownian control problem.  Let $B$ be standard Brownian motion with its completed natural filtration and set
\begin{equation}
    \Estar(\kappa)
    :=
    \inf\left\{
        \esssup_{0<t<1}\E u_t^2:
        B_1+\int_0^1u_t\,dt\ge\kappa\ \text{a.s.}
    \right\},
\label{eq:estar}
\end{equation}
where the infimum is over progressively measurable controls.  We refer to $\Estar(\kappa)$ as the Brownian control value.

Our main result is that this control problem gives the exact online threshold.

\begin{theorem}[Exact online threshold]
\label{thm:main}
For every fixed $\kappa\in\R$,
\begin{equation}
    \aon(\kappa)
    =
    \frac{2}{\pi\Estar(\kappa)}.
\label{eq:main-threshold}
\end{equation}
Moreover, if
\[
    \alpha
    <
    \frac{2}{\pi\Estar(\kappa)},
\]
then there is a constant $\rho>0$ and a deterministic online algorithm such that, for every integer sequence $M_N/N\to\alpha$,
\[
    \Pp\left(
        \min_{1\le i\le M_N}
        \frac{(G\sigma)_i}{\sqrt N}
        \ge\kappa+\rho
    \right)
    \longrightarrow1.
\]
This algorithm uses $O(MN)$ arithmetic operations.
\end{theorem}

The theorem is an exact characterization of the online threshold, rather than only an upper bound or an analysis of a particular algorithm.  Above the value in \eqref{eq:main-threshold}, no online algorithm can succeed with probability tending to one.  Below it, the proof gives a deterministic online algorithm.

A finite-precision implementation with polynomial bit complexity is given in \cref{app:finite-precision}.

\paragraph{Zero margin.}
At $\kappa=0$ we give rigorous numerical bounds on $\Estar(0)$.  The two sides use independent computer-assisted arguments.

\begin{theorem}[Zero margin]
\label{thm:zero}
One has
\begin{equation}
    1.73639
    <
    \Estar(0)
    \le\frac{243}{125}=1.944.
\label{eq:zero-E}
\end{equation}
Consequently,
\begin{equation}
    \frac{250}{243\pi}
    =
    0.3274793\ldots
    \le
    \aon(0)
    <
    \frac{200000}{173639\pi}
    =
    0.3666340\ldots.
\label{eq:zero-bracket}
\end{equation}
\end{theorem}

In particular, \cref{thm:main,thm:zero} give a deterministic online algorithm at every fixed density
\[
    \alpha<\frac{250}{243\pi}=0.3274793\ldots.
\]
For comparison, the classical multistage majority algorithm of Kim and Roche is proved to succeed at zero margin for $\alpha<0.005$~\cite{KimRoche98}, while Li, Schramm, and Zhou give an offline polynomial-time algorithm for $\alpha\le0.1$~\cite{LiSchrammZhou25}.  Thus a deterministic algorithm that must fix each sign before the later columns arrive still surpasses, by more than a factor of three, the best density previously proved attainable in polynomial time.  The storage capacity is at least as large, since it records only whether a satisfying $\sigma$ exists, which every successful online run exhibits.  Write $\alpha_\star(\kappa)$ for this storage capacity.  A recent computer-assisted proof of Shmalo, closing the outstanding numerical conditions in the matching bounds of Ding--Sun and Huang, gives
\[
    \alpha_\star(0)
    \in
    [0.833078599,0.833078600]
\]
\cite{DingSun25,Huang24,Shmalo26}.  Hence the exact online threshold is separated by a constant factor from the information-theoretic storage capacity.

\paragraph{Large margins.}
The Brownian characterization also gives sharp first-order behavior for large positive margin:
\begin{equation}
    \Estar(\kappa)
    =
    \kappa^2\bigl(1+O(\kappa^{-1})\bigr),
    \qquad
    \aon(\kappa)
    =
    \frac{2}{\pi\kappa^2}
    \bigl(1+O(\kappa^{-1})\bigr)
    \qquad
    (\kappa\to+\infty).
\label{eq:large-positive-intro}
\end{equation}
Li, Schramm, and Zhou give an offline polynomial-time construction with the same leading constant, while Stojnic's rigorous upper bound on the storage capacity has the same asymptotic form \cite{LiSchrammZhou25,Stojnic13}.  Since an algorithmic lower bound is also a lower bound on satisfiability, these two results imply that the offline storage capacity itself has leading term $2/(\pi\kappa^2)$.  Thus the online threshold and the offline storage capacity agree to first order as $\kappa\to+\infty$.

For large negative $\kappa$, the picture is different:
\[
    \aon(\kappa)
    =
    \Theta\!\left(
        \frac1{\kappa^2\Ph(\kappa)}
    \right)
    \qquad
    (\kappa\to-\infty),
\]
where $\Ph$ is the standard normal distribution function.  \Cref{thm:large-negative} gives the two-sided constants.

This is the same order as the best currently proved offline polynomial-time guarantee of Li, Schramm, and Zhou \cite{LiSchrammZhou25}.  In contrast, by \cite[Theorem~2.4]{LiSchrammZhou25} the offline storage capacity satisfies
\[
    \alpha_\star(\kappa)
    \sim
    \frac{\log 2}{\Ph(\kappa)}
    \qquad
    (\kappa\to-\infty),
\]
so it is larger than the online threshold by a factor of order $\kappa^2$.  Li, Schramm, and Zhou also obtain a multi-overlap-gap obstruction for stable algorithms at a scale larger by a $\log^2|\kappa|$ factor than their algorithmic guarantee.  The exact online threshold therefore singles out, without a computational assumption, the scale $1/(\kappa^2\Ph(\kappa))$ for algorithms constrained by sequential information.

\paragraph{Relation to online vector balancing.}
Online vector balancing and online discrepancy have been studied under several adversarial and stochastic input models \cite{BansalSpencer20,BansalJiangSinglaSinha20,GamarnikKizildagPerkinsXu23,KulkarniReisRothvoss24}, including a recent linear-time algorithm attaining the optimal prefix discrepancy \cite{AdenAli26}. Our starting point is the recent work of Fiedler, Jackson, Lacker, and Niles-Weed \cite{FiedlerJacksonLackerNilesWeed26}, who determine the mean-field limit of online stochastic vector balancing.  Their limiting problem asks for the narrowest terminal interval containing a controlled Brownian motion under a uniform-in-time $L^2$ constraint on the drift.  We use several components of their framework: the Doob decomposition and compactness argument identifying the Brownian limit, the sharp one-column Gaussian drift constraint and its realization by a sign, the Gaussian coupling controlling the maximum over coordinates, and the F\"ollmer drift construction.  Their results already allow the number of coordinates and the number of arrivals to grow proportionally, and we use that extension directly.

The main constructive difficulty is specific to the perceptron: the Brownian limit describes a typical coordinate, whereas success requires all $M=\Theta(N)$ inequalities to hold simultaneously.  The half-line terminal set $[\kappa,\infty)$ is what makes this tractable.  If a Brownian control is feasible, replacing its drift by its positive part preserves feasibility, and the terminal shortfall of the resulting nonnegative controls is dominated pointwise by the corresponding Brownian shortfall.  These two monotonicity properties let us approximate an arbitrary feasible Brownian control by simple predictable controls without assuming regularity of an optimizer, and a final-block construction then enforces all $M$ inequalities simultaneously.

\paragraph{Prior online results for the asymmetric perceptron.}
Closest to the present problem, Benedetti et al.~\cite{BenedettiBogdanovEtAl26} give an online algorithm for positive-tail activations.  For the ABP, its single-vector specialization finds a preimage of the all-positive output at sufficiently small constant density, with constant success probability, a guarantee subsumed by our exact online threshold.  Their main result instead constructs extensively separated colliding inputs, and their online lower bound concerns a different randomized non-periodic activation and the collision problem, rather than ABP feasibility.

\paragraph{Concurrent and independent work.}
In independent and concurrent work, Huang, Sellke, and Sun~\cite{HuangSellkeSun26} determine exact algorithmic thresholds for optimizing the Hamiltonian of the spherical and Ising perceptron, characterized, as here, by a one-dimensional stochastic control problem.  Their setting differs from ours: they treat offline algorithms with dimension-free Lipschitz dependence on the disorder, with hardness from the branching overlap gap property, whereas we choose signs online and ask for feasibility.  Neither result subsumes the other.

\paragraph{Proof overview.}
We briefly describe the two directions of \cref{thm:main}.  For impossibility above the threshold, decompose each signed Gaussian increment into its predictable drift and martingale difference.  The one-column drift bound implies a uniform empirical $L^2$ bound.  Following a uniformly chosen row and passing to a subsequence gives
\[
    B_t+\int_0^t a_s\,ds,
    \qquad
    \E a_t^2\le\frac{2}{\pi\alpha},
\]
with terminal value at least $\kappa$.  The limiting filtration may contain more information than the natural filtration of $B$; projecting the drift onto the Brownian filtration by conditional expectation preserves the half-line terminal inequality and can only decrease the second moment of the drift.  This gives $\Estar(\kappa)\le2/(\pi\alpha)$.

For achievability below the threshold, start from any feasible Brownian control with uniform $L^2$ norm strictly below $\ca=\sqrt{2/(\pi\alpha)}$.  After introducing a fixed amount of slack, the half-line monotonicity described above lets us approximate this control by a bounded simple predictable control depending on finitely many past Brownian values.  The one-column Gaussian construction realizes its coordinatewise drift using actual signs, and a conditional Gaussian coupling keeps the signed partial sums close to their Brownian counterparts in $\ell^\infty$.

To force all $M$ inequalities at once, we stop the first stage before time one and reserve the remaining block for coordinatewise correction.  At the start of that block, each coordinate with positive shortfall receives a deterministic drift equal to its shortfall divided by the block length.  A F\"ollmer drift associated with a compactly supported terminal law is superimposed to confine the fresh Gaussian fluctuation.  An $L^p$ bound with $p>4$ controls the empirical size of these coordinatewise drifts uniformly over the final block, while the conditional coupling makes the Brownian comparison simultaneous over all rows.  This yields a positive margin above $\kappa$ and proves achievability.

For the lower bound on $\Estar(0)$, we replace the uniform-in-time constraint by weighted integrated $L^2$ costs.  For piecewise-constant weights, the resulting stochastic control problem has an exact entropic representation and a scalar dynamic programming recursion.  A computer-assisted global enclosure of this recursion gives the lower endpoint in \eqref{eq:zero-E}.  The upper endpoint comes from an explicit self-similar diffusion: verified differential inequalities control the mean and variance of its drift, after which a deterministic correction makes the second moment constant in time while preserving the terminal half-line constraint.  For large negative margin, the upper bound uses a compactly supported terminal law and its F\"ollmer drift.  The lower bound reduces to a bivariate Gaussian tail on the correlation scale $1-t\asymp\kappa^{-2}$, which yields the explicit constants in the quantitative bound.

\paragraph{Organization.}
\Cref{sec:preliminaries} defines the online model and records the one-column Gaussian construction.  \Cref{sec:mean-field} proves the impossibility direction of \cref{thm:main}.  \Cref{sec:achievability} gives the online algorithm and proves achievability below $2/(\pi\Estar(\kappa))$. \Cref{sec:quantitative} proves \cref{thm:zero}, the large-positive-margin asymptotic, and \cref{thm:large-negative}.  The appendices contain the Brownian compactness argument, the conditional Gaussian coupling, the F\"ollmer and Euler estimates, the finite-precision analysis, and the computer-assisted verification.

    \section{Preliminaries}
\label{sec:preliminaries}

\subsection{Online algorithms and the constraint density}

Throughout the paper, $G\in\R^{M\times N}$ has i.i.d.\ standard Gaussian entries.  We write
\[
    G=(g_1,\ldots,g_N),
    \qquad
    g_k\in\R^M,
\]
and consider integer sequences $M=M_N$ such that
\[
    \frac{M}{N}\longrightarrow\alpha\in(0,\infty).
\]
As usual for the binary perceptron, we call $\alpha$ the constraint density.

The columns are revealed sequentially.  After observing $g_k$, an online algorithm chooses the sign $\sigma_k\in\{-1,+1\}$ before $g_{k+1}$ is revealed.  Let
\[
    \cF_k^G:=\sigma(g_1,\ldots,g_k),
    \qquad
    \cF_0^G:=\{\varnothing,\Omega\}.
\]

\begin{definition}[Online algorithm]
\label{def:online-algorithm-class}
A deterministic online algorithm is a sequence $\sigma_1,\ldots,\sigma_N\in\{-1,+1\}$ such that $\sigma_k$ is $\cF_k^G$-measurable for every $1\le k\le N$.  Its normalized partial sums are
\begin{equation}
    X_k
    :=
    \frac1{\sqrt N}
    \sum_{j=1}^k\sigma_jg_j,
    \qquad
    0\le k\le N,
\label{eq:partial-margin}
\end{equation}
with $X_0=0$.  The algorithm succeeds at margin $\kappa$ if
\[
    X_N\ge\kappa\1_M
\]
coordinatewise.
\end{definition}

\begin{definition}[Online threshold]
\label{def:online-threshold}
A constraint density $\alpha>0$ is \emph{online-achievable} at margin $\kappa$ if, for every integer sequence $M_N/N\to\alpha$, there is a sequence of online algorithms such that
\[
    \Pp\left(
        \min_{1\le i\le M_N}X_{N,i}\ge\kappa
    \right)
    \longrightarrow1.
\]
The online threshold is
\[
    \aon(\kappa)
    :=
    \sup\{
        \alpha>0:
        \alpha\text{ is online-achievable at margin }\kappa
    \}.
\]
\end{definition}

Allowing an independent random seed does not change this definition.  Indeed, if a randomized algorithm succeeds with probability $p_{M,N}$ over the matrix and the seed, then for each $(M,N)$ there is a fixed seed for which the conditional success probability is at least $p_{M,N}$.  We may therefore work with deterministic online algorithms throughout.

\subsection{The Brownian control problem}

Write
\[
    \ph(x)
    :=
    \frac1{\sqrt{2\pi}}e^{-x^2/2},
    \qquad
    \Ph(x)
    :=
    \int_{-\infty}^x\ph(y)\,dy
\]
for the standard normal density and distribution function.  Let $B=(B_t)_{0\le t\le1}$ be a standard Brownian motion with completed natural filtration $\cF^B=(\cF_t^B)_{0\le t\le1}$.  For an $\cF^B$-progressively measurable process $u$, define
\begin{equation}
    Y_t^u
    :=
    B_t+\int_0^tu_s\,ds,
    \qquad
    \norm{u}_{\infty,2}
    :=
    \esssup_{0<t<1}
    \bigl(\E u_t^2\bigr)^{1/2}.
\label{eq:brownian-controlled-path}
\end{equation}
Thus $\norm{u}_{\infty,2}$ is the mixed norm $L_t^\infty L_\omega^2$, i.e., $L^\infty((0,1);L^2(\Omega))$. We consider controls with $\norm{u}_{\infty,2}<\infty$; in particular, $\E\int_0^1u_t^2\,dt<\infty$.  The control value appearing in \cref{thm:main} is
\[
    \Estar(\kappa)
    =
    \inf\left\{
        \norm{u}_{\infty,2}^2:
        Y_1^u\ge\kappa\ \text{a.s.}
    \right\}.
\]

The following elementary estimates will also be used in \cref{sec:quantitative}.

\begin{lemma}
\label{lem:elementary-E-bounds}
For every $\kappa\in\R$,
\begin{equation}
    \E(\kappa-Z)_+^2
    \le
    \Estar(\kappa)
    \le
    \bigl(\kappa_++\sqrt8\bigr)^2,
    \qquad
    Z\sim N(0,1).
\label{eq:elementary-E-bounds}
\end{equation}
In particular,
\[
    0<\Estar(\kappa)<\infty.
\]
\end{lemma}

\begin{proof}
Let $u$ be feasible at level $\kappa$ and put
\[
    A:=\int_0^1u_t\,dt.
\]
On $\{B_1<\kappa\}$, feasibility gives $A\ge\kappa-B_1>0$, and hence
\[
    A^2
    \ge
    (\kappa-B_1)^2
    \1_{\{B_1<\kappa\}}.
\]
By Minkowski's inequality,
\[
    \norm{A}_{L^2}
    \le
    \int_0^1\norm{u_t}_{L^2}\,dt
    \le
    \norm{u}_{\infty,2}.
\]
Taking expectations and then the infimum over feasible controls proves the lower bound.

For the upper bound, let $R$ be a five-dimensional Bessel process started from $r_0=\sqrt2$:
\[
    dR_t=dW_t+\frac2{R_t}\,dt,
    \qquad
    R_0=r_0,
\]
where $W$ is Brownian motion.  The process $R$ is strictly positive. Applying It\^o's formula to $R^{-2}$ up to the usual localization times and then removing the localization gives
\[
    d(R_t^{-2})
    =
    -2R_t^{-3}\,dW_t
    -
    R_t^{-4}\,dt,
    \qquad
    \E R_t^{-2}\le\frac12.
\]
Therefore
\[
    R_t-r_0+r_0t
    =
    W_t+\int_0^t
    \left(
        r_0+\frac2{R_s}
    \right)ds
\]
has terminal value $R_1>0$ almost surely.  Its drift
\[
    v_t:=r_0+\frac2{R_t}
\]
satisfies
\[
    \E v_t^2
    \le
    2r_0^2+8\E R_t^{-2}
    \le8.
\]
Thus $v+\kappa_+$ is feasible at level $\kappa$, and Minkowski's inequality gives
\[
    \norm{v+\kappa_+}_{\infty,2}
    \le
    \sqrt8+\kappa_+.
\]
This proves the upper bound in \eqref{eq:elementary-E-bounds}.
\end{proof}

\subsection{Admissible drifts for one Gaussian column}

The discrete process also has a sharp one-step constraint on its predictable drift.  The following construction is the Gaussian one-column argument used in~\cite[Section~1.4 and Lemma~4.7]{FiedlerJacksonLackerNilesWeed26}; we state it in the $M\times N$ normalization used throughout this paper.

Let $\xi\sim N(0,1)$ and set
\begin{equation}
    \cgauss
    :=
    \E|\xi|
    =
    \sqrt{\frac2\pi},
    \qquad
    \ca
    :=
    \frac{\cgauss}{\sqrt\alpha}
    =
    \sqrt{\frac{2}{\pi\alpha}}.
\label{eq:steering-constants}
\end{equation}

\begin{lemma}
\label{lem:one-column-drift-bound}
Let $Z\sim N(0,I_M/N)$ be independent of a sigma-field $\cH$, and let $\varepsilon\in\{-1,+1\}$ be measurable with respect to $\cH\vee\sigma(Z)$.  Then
\[
    b
    :=
    \E[\varepsilon Z\mid\cH]
\]
satisfies
\[
    \norm{b}_2
    \le
    \frac{\cgauss}{\sqrt N}
    \qquad\text{a.s.}
\]
\end{lemma}

\begin{proof}
On $\{b\ne0\}$, the vector $v:=b/\norm{b}_2$ is $\cH$-measurable and has unit norm.  Hence
\[
    \norm{b}_2
    =
    \langle v,b\rangle
    =
    \E[
        \varepsilon\langle v,Z\rangle
        \mid\cH
    ]
    \le
    \E[
        |\langle v,Z\rangle|
        \mid\cH
    ]
    =
    \sqrt{\frac{2}{\pi N}}.
\]
The claim is trivial on $\{b=0\}$.
\end{proof}

Thus every conditional drift of one signed Gaussian column lies in the Euclidean ball of radius $\cgauss/\sqrt N$.  Conversely, every predictable vector in this ball can be realized exactly while retaining a fresh Gaussian innovation.

\begin{lemma}
\label{lem:one-column-steering}
Let $Z\sim N(0,I_M/N)$ be independent of a sigma-field $\cH$, and let $b\in\R^M$ be $\cH$-measurable with
\[
    \norm{b}_2
    \le
    \frac{\cgauss}{\sqrt N}.
\]
Then there are an $\cH\vee\sigma(Z)$-measurable sign $\varepsilon\in\{-1,+1\}$ and a random vector $\Delta B$ such that
\[
    \E[\varepsilon Z\mid\cH]=b,
    \qquad
    \Delta B\sim N(0,I_M/N),
\]
with $\Delta B$ independent of $\cH$.  Moreover, given $\cH$, the map
\[
    Z\longmapsto(\varepsilon,\Delta B)
\]
is invertible.

More explicitly, on $\{b\ne0\}$ set
\[
    v
    :=
    \frac{b}{\norm{b}_2},
    \qquad
    R
    :=
    \sqrt N\norm{b}_2\in[0,\cgauss],
    \qquad
    \Psi(R)
    :=
    \sqrt{
        -2\log\!\left(
            1-\frac{R}{\cgauss}
        \right)
    },
\]
with $\Psi(\cgauss)=+\infty$.  After observing $Z$, define
\begin{equation}
    \varepsilon
    :=
    \begin{cases}
    -1,
    &-\Psi(R)<\sqrt N\langle Z,v\rangle\le0,\\
    +1,
    &\text{otherwise},
    \end{cases}
\label{eq:steering-sign}
\end{equation}
and
\begin{equation}
    \Delta B
    :=
    \langle Z,v\rangle v
    +
    \varepsilon
    (I-vv^{\mathsf T})Z.
\label{eq:brownian-innovation}
\end{equation}
On $\{b=0\}$, take $\varepsilon=1$ and $\Delta B=Z$.
\end{lemma}

\begin{proof}
On $\{b\ne0\}$ let
\[
    \eta:=\sqrt N\langle Z,v\rangle\sim N(0,1).
\]
Conditionally on $\cH$, $\eta$ is standard normal and
\[
    \E\left[
        \eta
        \left(
            1-2\1_{\{-\Psi(R)<\eta\le0\}}
        \right)
        \,\middle|\,
        \cH
    \right]
    =
    \sqrt{\frac2\pi}
    \left(
        1-e^{-\Psi(R)^2/2}
    \right)
    =
    R.
\]
The component of $Z$ orthogonal to $v$ is centered Gaussian, independent of $\eta$.  Since $\varepsilon$ depends only on $\eta$, its orthogonal contribution has conditional mean zero.  Hence
\[
    \E[\varepsilon Z\mid\cH]
    =
    \frac{R}{\sqrt N}v
    =
    b.
\]

Conditionally on $\cH$, the parallel component $\langle Z,v\rangle v$ and the orthogonal component $(I-vv^{\mathsf T})Z$ are independent Gaussian vectors.  Multiplying the orthogonal component by the sign $\varepsilon$, which depends only on the parallel component, preserves its centered Gaussian law and its independence from the parallel component.  Thus \eqref{eq:brownian-innovation} is $N(0,I_M/N)$ conditionally on $\cH$, and therefore is independent of $\cH$.

Finally,
\[
    \langle\Delta B,v\rangle
    =
    \langle Z,v\rangle,
\]
so \eqref{eq:steering-sign} recovers $\varepsilon$ from $\Delta B$ and the past.  Then
\[
    Z
    =
    \langle\Delta B,v\rangle v
    +
    \varepsilon
    (I-vv^{\mathsf T})\Delta B.
\]
This proves invertibility.
\end{proof}

For the coordinatewise controls used in \cref{sec:achievability}, we scale an admissible one-column drift as $b=q/N$.  By \cref{lem:one-column-steering}, this is admissible exactly when $\norm{q}_2\le\cgauss\sqrt N$, equivalently $\big(\frac1M\sum_{i=1}^Mq_i^2\big)^{1/2}\le\ca+o(1)$ when $M/N\to\alpha$.

For later use we name this map: whenever $\norm{q}_2\le\cgauss\sqrt N$, write
\[
    (\varepsilon,\Delta B)
    =\operatorname{Steer}_N(q,Z)
\]
for the sign and innovation produced by \cref{lem:one-column-steering} with $b=q/N$.  In words, $\operatorname{Steer}_N$ realizes the prescribed coordinatewise drift $q/N$ with a single sign $\varepsilon$ while returning a fresh Gaussian innovation $\Delta B$.

The constant $\ca$ is therefore the limiting $L^2$ budget for the coordinatewise Brownian control in the proportional limit.  This is the normalization that appears in both the upper bound of \cref{sec:mean-field} and the algorithm of \cref{sec:achievability}.

    \section{Upper bound on the online threshold}
    \label{sec:mean-field}

    The upper bound follows from the one-column drift constraint together with a Brownian scaling limit.  The latter is a direct specialization of the compactness argument of Fiedler--Jackson--Lacker--Niles-Weed \cite[Section~3]{FiedlerJacksonLackerNilesWeed26}; their proportional extension is stated in~\cite[Equations~(1.3)--(1.5)]{FiedlerJacksonLackerNilesWeed26}. For the present feasibility problem, it is enough to follow a row chosen uniformly at random.  We therefore state only the one-dimensional limit needed here; a self-contained proof is given in \cref{app:mean-field}.

    \begin{proposition}
    \label{prop:mean-field-limit}
    Suppose that $M/N\to\alpha\in(0,\infty)$ and that a sequence of online algorithms satisfies
    \[
        \Pp\left(\min_{1\le i\le M}X_{N,i}\ge\kappa\right)\longrightarrow1.
    \]
    Then there exist a filtered probability space with filtration $\cG$, a $\cG$-Brownian motion $B$, and a $\cG$-progressively measurable process $a$ such that
    \begin{equation}
        B_1+\int_0^1 a_t\,dt\ge\kappa\quad\text{a.s.},
        \qquad
        \E a_t^2\le\frac{2}{\pi\alpha}
        \quad\text{for a.e. }t\in[0,1].
    \label{eq:mean-field-limit}
    \end{equation}
    \end{proposition}

    \begin{proof}[Proof sketch]
    Put $Z_k=g_k/\sqrt N$, and let $\cH_{k-1}$ denote the information available before the $k$th column is revealed.  The Doob decomposition of the signed increment is
    \[
        b_k:=\E[\sigma_k Z_k\mid\cH_{k-1}],
        \qquad
        d_k:=\sigma_k Z_k-b_k.
    \]
    By \cref{lem:one-column-drift-bound},
    \begin{equation}
        \norm{b_k}_2^2\le\frac{2}{\pi N}.
    \label{eq:upper-one-step-energy}
    \end{equation}
    Let $I_N$ be uniform on $\{1,\ldots,M\}$, independently of the matrix, and linearly interpolate the predictable and martingale parts of row $I_N$ from
    \[
        A_{I_N}^N(k/N)=\sum_{j=1}^k b_{j,I_N},
        \qquad
        W_{I_N}^N(k/N)=\sum_{j=1}^k d_{j,I_N}.
    \]
    On $((k-1)/N,k/N)$, averaging over the choice of $I_N$ gives
    \[
        \E\!\left[|\dot A_{I_N}^N(t)|^2\,\middle|\,b_k\right]
        =\frac{N^2}{M}\norm{b_k}_2^2
        \le\frac{2N}{\pi M}.
    \]
    Hence every subsequential limit has an absolutely continuous predictable part $A_t=\int_0^t a_s\,ds$ with $\E a_t^2\le2/(\pi\alpha)$ for almost every $t$.

    The centered increments satisfy
    \[
        \E[d_k\mid\cH_{k-1}]=0,
        \qquad
        \operatorname{Cov}(d_k\mid\cH_{k-1})
        =\frac1N I_M-b_kb_k^{\mathsf T},
    \]
    and have coordinatewise fourth moments $O(N^{-2})$.  The martingale functional central limit theorem, jointly with the predictable part, therefore identifies the second component of every subsequential limit as Brownian motion in the filtration generated by the limiting pair.  Finally,
    \[
        A_{I_N}^N(1)+W_{I_N}^N(1)=X_{N,I_N}.
    \]
    Since success forces every row above $\kappa$ with probability tending to one, the limiting terminal value is at least $\kappa$ almost surely.  This yields \eqref{eq:mean-field-limit}.  The compactness and filtration details are given in \cref{app:mean-field}.
    \end{proof}

    The filtration in \cref{prop:mean-field-limit} may be larger than the natural filtration of $B$.  The following projection reduces to the filtration used in the definition of $\Estar(\kappa)$.  The same conditioning argument appears in the proof of~\cite[Lemma~2.4]{FiedlerJacksonLackerNilesWeed26}.

    \begin{lemma}
    \label{lem:brownian-projection}
    Let $B$ be Brownian motion with respect to a filtration $\cG=(\cG_t)_{0\le t\le1}$, and let $a$ be $\cG$-progressively measurable with
    \[
        \E\int_0^1a_t^2\,dt<\infty.
    \]
    Suppose
    \[
        B_1+\int_0^1a_t\,dt\ge\kappa
        \qquad\text{a.s.}
    \]
    Let $\cF^B$ be the completed natural filtration of $B$, and choose an $\cF^B$-progressively measurable version of
    \[
        u_t=\E[a_t\mid\cF_t^B]
    \]
    for $dt\otimes d\Pp$-almost every $(t,\omega)$.  Then
    \begin{equation}
        B_1+\int_0^1u_t\,dt
        =\E\left[
            B_1+\int_0^1a_t\,dt
            \,\middle|\,
            \cF_1^B
        \right]
        \ge\kappa
        \qquad\text{a.s.},
    \label{eq:projection-terminal}
    \end{equation}
    and
    \[
        \E u_t^2\le\E a_t^2
        \qquad\text{for a.e. }t\in[0,1].
    \]
    \end{lemma}

    \begin{proof}
    Because $B$ is Brownian with respect to $\cG$, the future increments $(B_s-B_t)_{t\le s\le1}$ are independent of $\cG_t$.  Thus, for every integrable $\cG_t$-measurable random variable $H$,
    \[
        \E[H\mid\cF_1^B]=\E[H\mid\cF_t^B].
    \]
    Applying this identity to bounded truncations of $a_t$, followed by conditional Fubini and $L^1$ convergence, gives
    \[
        \E\left[
            \int_0^1a_t\,dt
            \,\middle|\,
            \cF_1^B
        \right]
        =\int_0^1u_t\,dt.
    \]
    Since $B_1$ is $\cF_1^B$-measurable, this proves the identity in \eqref{eq:projection-terminal}; the inequality follows by taking conditional expectations of the assumed terminal constraint.  Conditional Jensen gives
    \[
        \E u_t^2
        =\E\!\left[\E[a_t\mid\cF_t^B]^2\right]
        \le\E a_t^2
    \]
    for almost every $t$.
    \end{proof}

    \begin{corollary}
    \label{cor:threshold-upper}
    For every $\kappa\in\R$,
    \[
        \aon(\kappa)\le\frac{2}{\pi\Estar(\kappa)}.
    \]
    \end{corollary}

    \begin{proof}
    Let $\alpha$ be online-achievable at margin $\kappa$.  By \cref{prop:mean-field-limit}, there are $B$ and $a$ satisfying \eqref{eq:mean-field-limit}.  Applying \cref{lem:brownian-projection} gives a control adapted to the Brownian filtration, feasible at level $\kappa$, with
    \[
        \E u_t^2\le\frac{2}{\pi\alpha}
        \qquad\text{for a.e. }t.
    \]
    Therefore
    \[
        \Estar(\kappa)\le\frac{2}{\pi\alpha}.
    \]
    Every online-achievable constraint density $\alpha$ satisfies this inequality; taking the supremum over such $\alpha$ proves the claim.
    \end{proof}

    \section{Achievability and the online algorithm}
    \label{sec:achievability}

    We now prove the lower bound
    \[
        \aon(\kappa)\ge \frac{2}{\pi\Estar(\kappa)}
    \]
    and give the corresponding online algorithm.  Recall that
    \[
        \ca=\sqrt{\frac{2}{\pi\alpha}}.
    \]
    Thus the assumption $\alpha<2/(\pi\Estar(\kappa))$ is exactly $\ca^2>\Estar(\kappa)$: the Brownian control problem admits a feasible control whose uniform-in-time $L^2$ norm is strictly below the drift budget available in the discrete problem.

    Two points must be addressed before such a Brownian control can be used by an online algorithm.  First, a general progressively measurable control need not have a representation that can be evaluated from a discrete history.  We therefore approximate it by a simple predictable control whose coefficients are bounded Lipschitz functions of finitely many past Brownian values.  Second, Brownian feasibility concerns a single coordinate, whereas the perceptron requires all $M=\Theta(N)$ constraints to hold simultaneously.  We handle the latter issue by reserving a short final block: its deterministic component eliminates the shortfalls, while a F\"ollmer drift associated with a compactly supported terminal law confines the additional Gaussian fluctuation.

    \subsection{Approximation by simple predictable controls}
    \label{sec:approximation}

    The half-line constraint gives a monotonicity that lets us start from an arbitrary feasible Brownian control; no regularity of an optimizer for $\Estar(\kappa)$ is required.

    \begin{lemma}
    \label{lem:strict-slack}
    Let $\kappa\in\R$ and $c^2>\Estar(\kappa)$.  There exist $\rho>0$ and a nonnegative control $u$, progressively measurable with respect to the Brownian filtration, on $[0,1]$ such that
    \[
        \norm{u}_{\infty,2}<c,
        \qquad
        Y_1^u\ge \kappa+4\rho
        \quad\text{a.s.}
    \]
    Moreover, for every $T<1$ sufficiently close to one, the rescaled control
    \begin{equation}
        u_t^{(T)}:=T^{-1/2}u_{t/T},
        \qquad 0\le t\le T,
    \label{eq:compressed-control}
    \end{equation}
    satisfies
    \[
        \norm{u^{(T)}}_{\infty,2}<c,
        \qquad
        Y_T^{u^{(T)}}\ge \kappa+3\rho
        \quad\text{a.s.}
    \]
    \end{lemma}

    \begin{proof}
    Choose a feasible control $v$ with $\norm{v}_{\infty,2}<c$.  Replacing $v$ by $v_+:=\max\{v,0\}$ preserves feasibility and cannot increase $\norm{\cdot}_{\infty,2}$.  Choose $\eta>0$ so that $\norm{v_++\eta}_{\infty,2}<c$, set $u=v_++\eta$, and put $4\rho=\eta$.  Then $Y_1^u\ge\kappa+4\rho$ almost surely and $\norm{u}_{\infty,2}<c$.

    If $\widetilde B_t=\sqrt T B_{t/T}$, then $\widetilde B$ is Brownian motion on $[0,T]$ and
    \[
        \sqrt T\,Y_{t/T}^u
        =\widetilde B_t+\int_0^tT^{-1/2}u_{s/T}\,ds.
    \]
    Hence the control in \eqref{eq:compressed-control} has norm at most $\norm{u}_{\infty,2}/\sqrt T$, which is below $c$ for $T$ close enough to one.  Its terminal value is at least $\sqrt T(\kappa+4\rho)$, which is at least $\kappa+3\rho$ for all such $T$.
    \end{proof}

    We use the standard notion of a simple predictable process, specialized to deterministic time intervals.  For the algorithmic application we additionally choose its coefficients from an explicit dense class of functions of the Brownian history.

    \begin{definition}[Simple predictable control]
    \label{def:finite-observation}
    Let $0=s_0<\cdots<s_L=T$.  A bounded simple predictable control on $[0,T]$ is a process
    \begin{equation}
        \beta_t=\xi_j,
        \qquad t\in(s_j,s_{j+1}],
    \label{eq:finite-observation-control}
    \end{equation}
    where $\xi_j$ is bounded and $\cF_{s_j}^B$-measurable.  We will use controls for which, for each $j$,
    \[
        \xi_j
        =F_j(B_{r_{j,1}},\ldots,B_{r_{j,m_j}}),
        \qquad r_{j,\ell}\le s_j,
    \]
    with rational times $r_{j,\ell},s_j$ and bounded nonnegative globally Lipschitz rational piecewise-polynomial functions $F_j$.
    \end{definition}

    \begin{proposition}
    \label{prop:finite-observation-approximation}
    Suppose $u\ge0$ is progressively measurable with respect to the Brownian filtration on $[0,T]$, $\norm{u}_{\infty,2}<c$, and $Y_T^u\ge a$ almost surely.  Let $2\le p<\infty$ and $\varepsilon>0$.  Then there is a bounded nonnegative simple predictable control $\beta$ of the form in \cref{def:finite-observation} such that
    \begin{equation}
        \norm{\beta}_{\infty,2}<c,
        \qquad
        \norm{(a-Y_T^\beta)_+}_{L^p}<\varepsilon.
    \label{eq:finite-observation-output}
    \end{equation}
    \end{proposition}

    \begin{proof}
    Let $\mathsf H=L^2_{\mathrm{pred}}([0,T]\times\Omega)$ and let $\pi_n$ be the dyadic partition of $[0,T]$.  Let $\mathsf H_n$ be the closed subspace of processes that are constant on each interval $(t_{j,n},t_{j+1,n}]$ with an $\cF_{t_{j,n}}^B$-measurable coefficient.  The union of the $\mathsf H_n$ is dense in $\mathsf H$.  The orthogonal projection $P_nu$ has coefficient
    \[
        \xi_{j,n}
        :=
        \E\left[
           \frac1{\Delta_{j,n}}
           \int_{t_{j,n}}^{t_{j+1,n}}u_s\,ds
           \,\middle|\,\cF_{t_{j,n}}^B
        \right]
    \]
    on $(t_{j,n},t_{j+1,n}]$.  Since $u\ge0$, each $\xi_{j,n}$ is nonnegative, and conditional expectation followed by Minkowski's inequality gives
    \[
        \norm{\xi_{j,n}}_{L^2}
        \le
        \frac1{\Delta_{j,n}}
        \int_{t_{j,n}}^{t_{j+1,n}}\norm{u_s}_{L^2}\,ds
        \le\norm{u}_{\infty,2}.
    \]
    Moreover, $P_nu\to u$ in $\mathsf H$.

    For each coefficient $\xi_{j,n}$, bounded Lipschitz functions of finitely many rational-time Brownian evaluations are dense in $L^2(\cF_{t_{j,n}}^B)$.  After truncation, one first approximates the relevant sigma-field by Brownian values at rational times and then approximates the resulting function under the corresponding Gaussian law.  Tapering outside a rational box and triangulating a rational grid gives bounded rational piecewise-affine approximants; taking positive parts preserves nonnegativity.  Choosing these approximations sufficiently accurately and then letting $n\to\infty$ yields controls $\beta^{(n)}$ of the form in \cref{def:finite-observation} with
    \[
        \norm{\beta^{(n)}}_{\infty,2}<c,
        \qquad
        \norm{\beta^{(n)}-u}_{L^2([0,T]\times\Omega)}\longrightarrow0.
    \]

    Couple the controlled processes with the same Brownian motion.  Since $Y_T^u\ge a$,
    \[
        0\le(a-Y_T^{\beta^{(n)}})_+
        \le \abs{Y_T^{\beta^{(n)}}-Y_T^u}
        \le \int_0^T\abs{\beta_t^{(n)}-u_t}\,dt,
    \]
    so the shortfall converges to zero in $L^2$.  On the other hand, $\beta^{(n)}\ge0$ implies
    \begin{equation}
        (a-Y_T^{\beta^{(n)}})_+\le(a-B_T)_+.
    \label{eq:shortfall-domination}
    \end{equation}
    The right-hand side belongs to $L^r$ for every finite $r$.  Interpolation between the $L^2$ convergence and this uniform moment bound gives convergence in $L^p$, proving \eqref{eq:finite-observation-output}.
    \end{proof}

    \begin{remark}[Use of the half-line constraint]
    \label{rem:one-sidedness}
    The half-line constraint enters twice.  First, $u\mapsto u_+$ preserves feasibility.  Second, nonnegative approximants satisfy \eqref{eq:shortfall-domination}.  These two monotonicity properties allow us to approximate an arbitrary feasible Brownian control without assuming any regularity of an optimizer for $\Estar(\kappa)$.
    \end{remark}

    \subsection{The F\"ollmer drift on the final block}

    The first stage leaves a shortfall that is small in $L^p$, but this does not by itself control the minimum over $M=\Theta(N)$ coordinates.  On the final block we add a deterministic drift that compensates the shortfall of each coordinate.  The remaining Gaussian fluctuation is confined by a F\"ollmer drift associated with a compactly supported terminal law.

    \begin{lemma}[Final-block confinement]
    \label{lem:euler-follmer}
    Fix an even integer $p>4$ and $c_F>0$.  There are $h_0=h_0(p,c_F)>0$ and a radius $s_p(\cdot,c_F)$, continuous on $(0,h_0)$ and satisfying $s_p(h,c_F)\downarrow0$ as $h\downarrow0$, with the following property.  Let $L_N\in\{1,\ldots,N-1\}$, put $h_N=L_N/N$, and suppose that $h_N$ lies in a fixed compact subinterval of $(0,h_0)$.  There is an explicit feedback $v_N$ on $[0,h_N)$ such that, for independent Brownian coordinates, the Euler scheme
    \begin{equation}
    \widehat Z_{0,i}=0,
    \qquad
    \widehat a_{k,i}=v_N(k/N,\widehat Z_{k,i}),
    \qquad
    \widehat Z_{k+1,i}
    =\widehat Z_{k,i}+\frac{\widehat a_{k,i}}{N}
    +B_i((k+1)/N)-B_i(k/N)
    \label{eq:euler-follmer}
    \end{equation}
    satisfies the per-step drift bound
    \begin{equation}
        \sup_{k<L_N}\norm{\widehat a_{k,i}}_{L^p}\le c_F
    \label{eq:euler-request-bound}
    \end{equation}
    and the terminal confinement estimate
    \begin{equation}
    \E\max_{i\le M}
    \dist\left(
    \widehat Z_{L_N,i},
    [-s_p(h_N,c_F),s_p(h_N,c_F)]
    \right)
    =O(M^{1/p}N^{-1}).
    \label{eq:euler-endpoint-bound}
    \end{equation}
    In particular, if $M=O(N)$ the right-hand side is $o(1)$.
    \end{lemma}

    The feedback $v_N$ comes from the F\"ollmer construction associated with a compactly supported terminal law and is evaluated by an explicit ratio of truncated Gaussian moments, so it is efficiently computable and no continuous-time SDE solver is used.  The construction of $v_N$ and the proof of \cref{lem:euler-follmer} are given in \cref{app:follmer}.

\subsection{The online algorithm}

The algorithm maintains the normalized margin vector $X$ together with the Brownian state generated by the steering map.  A vector $q\in\R^M$ represents the coordinatewise Brownian control during one time step; the corresponding conditional drift of the signed column is $q/N$.

Define the projection onto the admissible control ball by
\begin{equation}
    \Pi_N(q):=
    \begin{cases}
    q,&\norm{q}_2\le\cgauss\sqrt N,\\[1mm]
    \displaystyle
    \frac{\cgauss\sqrt N}{\norm{q}_2}\,q,
    &\norm{q}_2>\cgauss\sqrt N.
    \end{cases}
\label{eq:request-projection}
\end{equation}
By \cref{lem:one-column-steering}, $\Pi_N(q)/N$ is an admissible predictable drift for the next signed Gaussian column.

Fix an even integer $p>4$.  The parameters $T,a,\beta,c_F$ used below are chosen before the column stream begins.  For each $N$, let $v_N$ denote the final-block feedback in \cref{lem:euler-follmer} with horizon $h_N=(N-\lfloor TN\rfloor)/N$ and budget $c_F$.  When evaluating $\beta$, every Brownian observation time in its fixed representation is rounded down to the $1/N$ mesh, and only the required Brownian mesh values are stored.

\begin{algorithm}[t]
\caption{\textsc{OnlinePerceptron}}
\label{alg:online}
\begin{algorithmic}[1]

\Statex \textbf{Parameters:}
rational $T\in(0,1)$, target $a$, simple predictable control $\beta$
on $[0,T]$, even $p>4$, and final-block budget $c_F$
\Statex \textbf{Input:}
$M,N$ and columns $g_1,\ldots,g_N\in\R^M$, revealed one at a time
\Statex \textbf{Output:}
signs $\sigma_1,\ldots,\sigma_N\in\{\pm1\}$

\State $k_0\gets\lfloor TN\rfloor$;
       $h_N\gets(N-k_0)/N$;
       $X,B\gets0\in\R^M$

\Statex
\Statex \textbf{Stage I: Follow the Brownian control}
\For{$k=1,\ldots,k_0$}
    \State Evaluate $\beta$ coordinatewise from the stored Brownian mesh
           values to obtain $q^{\rm raw}$
    \State $q\gets\Pi_N(q^{\rm raw})$
    \State Observe $g_k$ and set $Z\gets g_k/\sqrt N$
    \State $(\sigma_k,\Delta B)\gets\operatorname{Steer}_N(q,Z)$
    \State $X\gets X+\sigma_kZ$;
           $B\gets B+\Delta B$
    \State Store any mesh value of $B$ required later by $\beta$
\EndFor

\Statex
\Statex \textbf{Stage II: Enforce the remaining constraints}
\State $q_i^{\rm det}\gets(a-X_i)_+/h_N$ for every $i\le M$;
       $\widehat Z\gets0\in\R^M$
\For{$\ell=0,\ldots,N-k_0-1$}
    \State $\widehat a_i\gets v_N(\ell/N,\widehat Z_i)$ for every $i\le M$
    \State $q\gets\Pi_N(q^{\rm det}+\widehat a)$
    \State Observe $g_{k_0+\ell+1}$ and set
           $Z\gets g_{k_0+\ell+1}/\sqrt N$
    \State $(\sigma_{k_0+\ell+1},\Delta B)
           \gets\operatorname{Steer}_N(q,Z)$
    \State $X\gets X+\sigma_{k_0+\ell+1}Z$
    \State $\widehat Z\gets\widehat Z+\widehat a/N+\Delta B$
\EndFor

\State \Return $(\sigma_1,\ldots,\sigma_N)$

\end{algorithmic}
\end{algorithm}

In both stages the control vector is completely determined and projected before the next column is observed.  Thus the online information constraint is visible directly from the pseudocode.  Both stages have the same per-arrival structure: compute a predictable control, project it into the admissible ball, apply the Gaussian steering map, and update the state.  The two stages differ only in how the control vector is chosen.

For fixed parameters, each iteration uses $O(M)$ arithmetic operations, so \cref{alg:online} uses $O(MN)$ arithmetic operations in total.  Since the first-stage control depends on only finitely many Brownian observations, the working memory is $O(M)$ for every fixed choice of parameters.  A finite-precision implementation with polynomial bit complexity is given in \cref{app:finite-precision}.

    \subsection{Analysis of the online algorithm}

    We compare the signed process produced by the algorithm with the Brownian systems used to define its two controls.  The required estimate is the conditional $M$-coordinate, $N$-arrival form of the coupling from \cite[Lemma~4.7]{FiedlerJacksonLackerNilesWeed26}.  We state it explicitly because we use it twice---once for the first stage and once for the final block---and this second application is conditioned on the entire first stage.

    Let $t_k=k/N$ and define the Brownian mesh filtration
    \begin{equation}
        \cG_k^B:=\sigma(B_{t_0},\ldots,B_{t_k}).
    \label{eq:mesh-filtration}
    \end{equation}

    \begin{proposition}[Conditional Gaussian coupling]
    \label{prop:conditional-coupling}
    Let $B$ be an $M$-dimensional Brownian motion, and let $\cH_0$ be a sigma-field independent of its future increments.  Suppose
    \[
        b_k\in L^0(\cH_0\vee\cG_k^B;\R^M),
        \qquad
        \norm{b_k}_2\le\frac{\cgauss}{\sqrt N},
        \qquad 0\le k<N.
    \]
    Iterating the inverse map in \cref{lem:one-column-steering} produces independent vectors $Z_{k+1}\sim N(0,I_M/N)$ and online signs $\varepsilon_{k+1}$ such that
    \[
        \E[\varepsilon_{k+1}Z_{k+1}\mid\cH_0\vee\cG_k^B]=b_k.
    \]
    After adjoining $\cH_0$, the completed sigma-fields generated by $Z_1,\ldots,Z_k$ and by $B_{t_0},\ldots,B_{t_k}$ agree.

    For every $p>2$ there is a constant $C_p$, independent of $M$ and $N$, such that for $0\le k_0<k_1\le N$,
    \begin{align}
    &\E\left[
    \left\|
    \sum_{k=k_0}^{k_1-1}(\varepsilon_{k+1}Z_{k+1}-b_k)
    -(B_{t_{k_1}}-B_{t_{k_0}})
    \right\|_\infty
    \,\middle|\,\cH_0
    \right]
    \notag\\
    &\hspace{2cm}\le
    C_pN^{1/4}
    \max_{k_0\le k<k_1}
    \left(
    \sum_{i=1}^M\E[|b_{k,i}|^p\mid\cH_0]
    \right)^{1/(2p)}
    \qquad\text{a.s.}
    \label{eq:conditional-coupling}
    \end{align}
    \end{proposition}

    \begin{proof}
    The equality of the discrete filtrations follows by induction from the invertibility in \cref{lem:one-column-steering}.  For the maximal estimate, the proof of~\cite[Lemma~4.7]{FiedlerJacksonLackerNilesWeed26} represents the one-step discrepancy between the signed Gaussian increment and the Brownian increment as a scalar martingale difference in the current drift direction, and bounds its conditional moments.  The same coordinate martingale estimate, with $M$ coordinates and $N$ arrivals, gives \eqref{eq:conditional-coupling}; conditioning throughout on $\cH_0$ gives the stated conditional form.  The calculation is recorded in \cref{app:coupling}.
    \end{proof}

    The next proposition isolates the probabilistic estimates needed for the algorithm.  Its hypotheses are exactly the inequalities that will be arranged from the Brownian control problem.

    \begin{proposition}
    \label{prop:finite-transfer}
    Fix a constraint density $\alpha>0$ and an even integer $p>4$.  Let $a>r$ and $c_F>0$, let $T=1-h$ with $0<h<h_0(p,c_F)$, and let $\beta$ be a bounded simple predictable control on $[0,T]$ of the form in \cref{def:finite-observation} with
    \begin{equation}
        \norm{\beta}_{\infty,2}<\ca,
        \qquad
        \frac1h\norm{(a-Y_T^\beta)_+}_{L^p}+c_F<\ca,
        \qquad
        s_p(h,c_F)<a-r.
    \label{eq:transfer-hypotheses}
    \end{equation}
    Then \cref{alg:online}, with first-stage control $\beta$, target $a$, and final-block budget $c_F$, satisfies
    \begin{equation}
        \Pp\left(\min_{i\le M}X_{N,i}\ge r\right)\longrightarrow1
    \label{eq:finite-transfer-conclusion}
    \end{equation}
    whenever $M/N\to\alpha$.
    \end{proposition}

    \begin{proof}
    Put
    \[
        k_0:=\lfloor TN\rfloor,
        \qquad
        h_N:=\frac{N-k_0}{N}=h+O(N^{-1}).
    \]
    The projection in \eqref{eq:request-projection} makes every step admissible. We show that, with probability $1-o(1)$, it is never active.

    By the strict inequalities in \eqref{eq:transfer-hypotheses} we may fix constants $c_\beta,c_D>0$ with
    \[
        \norm{\beta}_{\infty,2}<c_\beta<\ca,
        \qquad
        \frac1h\norm{(a-Y_T^\beta)_+}_{L^p}<c_D,
        \qquad
        c_D+c_F<\ca.
    \]

    \paragraph{Stage I.}
    Apply the functions defining $\beta$ separately to the $M$ Brownian coordinate paths.  The collection of observation and block times is fixed. Rounding these times down to the $1/N$ mesh gives
    \begin{equation}
        \max_{i\le M}\max_{t\in\mathcal R}
        |B_i(t)-B_i(\lfloor Nt\rfloor/N)|
        =O_{\Pp}\!\left(\sqrt{\frac{\log N}{N}}\right),
    \label{eq:rounded-brownian-times}
    \end{equation}
    where $\mathcal R$ denotes the collection of times appearing in the representation of $\beta$.  Since the coefficient functions are bounded and globally Lipschitz, the discretized Brownian system at time $k_0/N$ differs in $\ell^\infty$ from $M$ independent copies of $Y_T^\beta$ by $o_{\Pp}(1)$.

    At each first-stage step, the coordinates of the unprojected control vector are bounded and independent across rows.  Since $\norm{\beta}_{\infty,2}<c_\beta$, concentration of the empirical second moment, together with a union bound over the control intervals, gives
    \begin{equation}
        \frac{\norm{q_k^{\rm raw}}_2}{\sqrt M}
        <c_\beta+o(1)
        <\cgauss\sqrt{\frac NM}
    \label{eq:first-stage-capacity}
    \end{equation}
    uniformly over the first stage, with probability $1-o(1)$.  Thus the projection is inactive there.

    Let $q_k:=\Pi_N(q_k^{\rm raw})$ be the control vector actually used by the algorithm, and let $\widetilde Y^N$ denote the Brownian mesh process driven by $q_k/N$.  The concentration estimate above shows that the projection is inactive throughout the first stage with probability $1-o(1)$; on that event, the preceding discretization estimate gives
    \[
        \norm{\widetilde Y^N-Y_T^{\beta,\otimes M}}_\infty=o_{\Pp}(1),
    \]
    where $Y_T^{\beta,\otimes M}$ denotes a vector of independent copies of $Y_T^\beta$.  Let $K$ bound the scalar values of $\beta$.  Since radial projection can only decrease each coordinate in absolute value, $|q_{k,i}|\le |q_{k,i}^{\rm raw}|\le K$.  We may therefore apply \cref{prop:conditional-coupling} to the always-admissible drift $b_k=q_k/N$, obtaining
    \[
        \sum_{i=1}^M\norm{b_{k,i}}_{L^p}^p
        \le MK^pN^{-p}
        =O(N^{1-p}),
    \]
    and therefore
    \begin{equation}
        \E\norm{X_{k_0}-\widetilde Y^N}_\infty
        =O\!\left(N^{-1/4+1/(2p)}\right)=o(1).
    \label{eq:first-stage-coupling}
    \end{equation}
    Consequently,
    \begin{equation}
        \norm{X_{k_0}-Y_T^{\beta,\otimes M}}_\infty=o_{\Pp}(1).
    \label{eq:first-stage-terminal-coupling}
    \end{equation}

    Choose $\delta>0$ such that
    \[
        \norm{(a-Y_T^\beta)_+}_{L^p}
        \le h(c_D-3\delta).
    \]
    The law of large numbers for the independent Brownian coordinates gives
    \[
        \left(
           \frac1M\sum_{i=1}^M(a-Y_{T,i}^\beta)_+^p
        \right)^{1/p}
        \le h(c_D-2\delta)
    \]
    with probability tending to one.  Since $x\mapsto(a-x)_+$ is $1$-Lipschitz, the $\ell^\infty$ coupling above transfers this estimate to $X_{k_0}$.  Using $h_N=h+O(N^{-1})$, we obtain
    \begin{equation}
        \left(
           \frac1M\sum_{i=1}^M(a-X_{k_0,i})_+^p
        \right)^{1/p}
        <h_N(c_D-\delta)
    \label{eq:empirical-shortfall}
    \end{equation}
    with probability $1-o(1)$.

    \paragraph{Stage II.}
    At time $k_0$ the algorithm sets
    \begin{equation}
        q_i^{\rm det}:=\frac{(a-X_{k_0,i})_+}{h_N}.
    \label{eq:deterministic-lift}
    \end{equation}
    On the event \eqref{eq:empirical-shortfall}, the empirical $L^p$ norm of $q^{\rm det}$ is below $c_D-\delta$, and hence so is its empirical $L^2$ norm.

    Conditionally on the first-stage Brownian mesh sigma-field, the future Brownian coordinates are independent.  The discretized F\"ollmer drifts are therefore independent across rows and satisfy the uniform $p$th-moment bound in \eqref{eq:euler-request-bound}.  Rosenthal's inequality \cite{Rosenthal70}, applied to $\widehat a_{\ell,i}^2$, followed by a union bound over the $O(N)$ final-stage times, gives
    \[
        \sup_\ell
        \left|
           \frac1M\sum_{i=1}^M\widehat a_{\ell,i}^2
           -\E\widehat a_{\ell,1}^2
        \right|
        =o_{\Pp}(1),
    \]
    because $p>4$ and $M\asymp N$.  Hence, uniformly over the final stage,
    \[
        \left(
           \frac1M\sum_{i=1}^M\widehat a_{\ell,i}^2
        \right)^{1/2}
        \le c_F+o_{\Pp}(1).
    \]
    Together with $c_D+c_F<\ca$ and Minkowski's inequality, this gives
    \begin{equation}
        \sup_\ell\frac{\norm{q_\ell^{\rm raw}}_2}{\sqrt M}
        \le c_D+c_F+o_{\Pp}(1)<\ca
    \label{eq:final-stage-capacity}
    \end{equation}
    with probability tending to one.  Thus the projection is inactive there as well.

    Ignoring the signing--Brownian coupling error, the terminal value of coordinate $i$ is
    \begin{align}
        X_{k_0,i}+h_Nq_i^{\rm det}+\widehat Z_{L_N,i}
        &=X_{k_0,i}+(a-X_{k_0,i})_++\widehat Z_{L_N,i}
        \notag\\
        &\ge a+\widehat Z_{L_N,i}.
    \label{eq:ideal-repaired-coordinate}
    \end{align}
    By \cref{lem:euler-follmer} and the continuity of $s_p(\cdot,c_F)$, the strict inequality $s_p(h,c_F)<a-r$ gives $s_p(h_N,c_F)<a-r$ for all large $N$, since $h_N=h+O(N^{-1})$.  Hence the minimum of this Brownian comparison process is larger than $r$ with probability tending to one.

    It remains to control the signing--Brownian coupling error in the final stage.  Let
    \[
        q_\ell:=\Pi_N(q^{\rm det}+\widehat a_\ell)
    \]
    be the control vector actually used at final-stage time $\ell$.  Condition on the first-stage history and apply \cref{prop:conditional-coupling} to the always-admissible drift $b_\ell=q_\ell/N$, with $\cH_0=\cG_{k_0}^B$.  Radial projection does not increase coordinatewise absolute values, so \eqref{eq:empirical-shortfall}, \eqref{eq:euler-request-bound}, and conditional Minkowski give
    \[
        \sum_{i=1}^M
        \E\left[
           \left|
              \frac{q_{\ell,i}}{N}
           \right|^p
           \,\middle|\,\cG_{k_0}^B
        \right]
        \le C N^{1-p}.
    \]
    On the event \eqref{eq:final-stage-capacity}, which has probability $1-o(1)$, the projection is inactive and this Brownian comparison is exactly the ideal process in \eqref{eq:ideal-repaired-coordinate}.  Thus the conditional expected $\ell^\infty$ coupling error is $O(N^{-1/4+1/(2p)})=o(1)$.  Combining this with \eqref{eq:ideal-repaired-coordinate} and \eqref{eq:euler-endpoint-bound} proves \eqref{eq:finite-transfer-conclusion}.
    \end{proof}

    We can now instantiate the preceding proposition from any Brownian control whose $L^2$ norm lies strictly below the available drift budget.

    \begin{proposition}
    \label{prop:achievability}
    If
    \[
        \alpha<\frac{2}{\pi\Estar(\kappa)},
    \]
    then there exists $\rho>0$ such that \cref{alg:online}, for a suitable fixed choice of its control parameters, satisfies
    \[
        \Pp\left(
           \min_{i\le M}X_{N,i}\ge\kappa+\rho
        \right)\longrightarrow1
    \]
    whenever $M/N\to\alpha$.
    \end{proposition}

    \begin{proof}
    The assumption is $\ca^2>\Estar(\kappa)$.  Fix an even integer $p>4$, apply \cref{lem:strict-slack} with $c=\ca$, and fix $c_F\in(0,\ca)$.  Choose a rational $T=1-h$ with $0<h<h_0(p,c_F)$ close enough to one that the rescaled control $u^{(T)}$ satisfies $\norm{u^{(T)}}_{\infty,2}<\ca$ and $Y_T^{u^{(T)}}\ge\kappa+3\rho_0$ almost surely for some $\rho_0>0$, and $s_p(h,c_F)<\rho_0$.  Fix $\varepsilon>0$ with $\varepsilon/h+c_F<\ca$, and apply \cref{prop:finite-observation-approximation} to $u^{(T)}$ with $c=\ca$, target $a=\kappa+3\rho_0$, and error $\varepsilon$, obtaining a simple predictable control $\beta$ with
    \[
        \norm{\beta}_{\infty,2}<\ca,
        \qquad
        \norm{(\kappa+3\rho_0-Y_T^\beta)_+}_{L^p}<\varepsilon.
    \]
    Set $a=\kappa+3\rho_0$ and $r=\kappa+2\rho_0$.  Then $\norm{\beta}_{\infty,2}<\ca$, $\tfrac1h\norm{(a-Y_T^\beta)_+}_{L^p}+c_F<\varepsilon/h+c_F<\ca$, and $s_p(h,c_F)<\rho_0=a-r$, so $\beta$ meets the hypotheses of \cref{prop:finite-transfer}.  Applying it and taking $\rho=2\rho_0$ proves the claim.
    \end{proof}

    Combining \cref{cor:threshold-upper,prop:achievability} gives
    \[
        \aon(\kappa)=\frac{2}{\pi\Estar(\kappa)},
    \]
    which completes the proof of the threshold identity in \cref{thm:main}. The implementation with polynomial bit complexity is proved in \cref{app:finite-precision}.
    \section{Quantitative analysis of the Brownian control value}
    \label{sec:quantitative}

    We now study $\Estar(\kappa)$ quantitatively.  At zero margin we obtain rigorous computer-assisted upper and lower bounds by two independent arguments.  The lower bound comes from a weighted $L^2$ relaxation and a one-dimensional dynamic-programming recursion; the upper bound comes from self-similar diffusion constructed below.  We then derive the large-positive-margin asymptotic from elementary estimates and analyze the large-negative-margin limit using a localized F\"ollmer drift and a Gaussian tail estimate near correlation one.

    \subsection{Zero margin}
    \label{sec:zero-margin}

    \paragraph{A weighted $L^2$ relaxation.}
    The definition of $\Estar(\kappa)$ constrains the second moment of the drift uniformly in time.  Replacing this by a weighted time average gives a family of lower bounds.  The corresponding unweighted integrated $L^2$ problem appears in \cite[Proposition~6.5 and Remark~6.6]{FiedlerJacksonLackerNilesWeed26}; we use a time-dependent weight.

    Let $\eta:(0,1)\to(0,\infty)$ be measurable with
    \[
        \int_0^1\eta(t)\,dt=1.
    \]
    Define
    \begin{equation}
        V_\kappa(\eta)
        :=
        \inf_{Y_1^u\ge\kappa\ \mathrm{a.s.}}
        \E\int_0^1\eta(t)u_t^2\,dt,
    \label{eq:weighted-value}
    \end{equation}
    where the infimum is over controls progressively measurable with respect to the Brownian filtration and with finite integrated second moment.

    \begin{lemma}
    \label{lem:weighted-lower-bound}
    For every such weight $\eta$,
    \[
        \Estar(\kappa)\ge V_\kappa(\eta).
    \]
    \end{lemma}

    \begin{proof}
    For every feasible control $u$,
    \[
        \E\int_0^1\eta(t)u_t^2\,dt
        \le
        \left(\esssup_{0<t<1}\E u_t^2\right)
        \int_0^1\eta(t)\,dt
        =
        \norm{u}_{\infty,2}^2.
    \]
    Taking the infimum over feasible controls proves the claim.
    \end{proof}

    For piecewise-constant weights, the relaxed problem reduces to a one-dimensional recursion.  We first record the single-interval variational identity used in the recursion.

    \begin{lemma}[Entropic representation]
    \label{lem:entropic-identity}
    Let $a>0$ and $\Delta>0$, and let $\Gamma:\R\to(-\infty,+\infty]$ be Borel measurable and bounded below. For Brownian motion started from $x$, define
    \[
        \mathcal V(x)
        :=
        \inf_u
        \E\left[
            a\int_0^\Delta u_t^2\,dt
            +
            \Gamma\left(
                x+B_\Delta+\int_0^\Delta u_t\,dt
            \right)
        \right],
    \]
    where the infimum is over progressively measurable controls satisfying $\E\int_0^\Delta u_t^2\,dt<\infty$.  Then
    \begin{equation}
        \mathcal V(x)
        =
        -2a\log
        P_\Delta\!\left(e^{-\Gamma/(2a)}\right)(x),
    \label{eq:entropic-identity}
    \end{equation}
    with the convention $e^{-\infty}=0$.
    \end{lemma}

    \begin{proof}
    Let $\mathsf W_x$ denote Wiener measure on $C([0,\Delta])$ started from $x$ and put
    \[
        Z
        :=
        \E_{\mathsf W_x}
        e^{-\Gamma(X_\Delta)/(2a)}
        =
        P_\Delta\!\left(e^{-\Gamma/(2a)}\right)(x).
    \]
    The Gibbs variational principle (equivalently, the entropy form of the Bou\'e--Dupuis variational representation) gives~\cite{BoueDupuis98}
    \begin{equation}
        -2a\log Z
        =
        \inf_{Q\ll\mathsf W_x}
        \left\{
            2a\,H(Q\mid\mathsf W_x)
            +
            \E_Q\Gamma(X_\Delta)
        \right\},
    \label{eq:gibbs-terminal-cost}
    \end{equation}
    where $H(\,\cdot\,\mid\,\cdot\,)$ denotes relative entropy.

    If $u$ is admissible and $Q^u$ is the law of $x+B+\int_0^\cdot u_s\,ds$, the localized Girsanov inequality yields
    \[
        H(Q^u\mid\mathsf W_x)
        \le
        \frac12\E\int_0^\Delta u_t^2\,dt.
    \]
    Substituting $Q^u$ into \eqref{eq:gibbs-terminal-cost} shows that the right-hand side of \eqref{eq:entropic-identity} is a lower bound for $\mathcal V(x)$.

    Suppose now that $Z>0$ and define
    \[
        \frac{dQ^\star}{d\mathsf W_x}
        =
        \frac{e^{-\Gamma(X_\Delta)/(2a)}}{Z}.
    \]
    Since $\Gamma$ is bounded below, the density is bounded.  Moreover $\Gamma(X_\Delta)$ is $Q^\star$-integrable: on the set where $\Gamma$ is finite, the positive part of $\Gamma e^{-\Gamma/(2a)}$ is bounded, while the negative part is bounded by assumption.  Define
    \[
        q_t(y)
        :=
        P_{\Delta-t}\!\left(e^{-\Gamma/(2a)}\right)(y),
        \qquad
        u_t^\star
        :=
        \partial_y\log q_t(X_t).
    \]
    The F\"ollmer drift representation gives
    \[
        X_t
        =
        x+W_t+\int_0^t u_s^\star\,ds,
        \qquad
        H(Q^\star\mid\mathsf W_x)
        =
        \frac12
        \E_{Q^\star}\int_0^\Delta (u_t^\star)^2\,dt,
    \]
    where $W$ is $Q^\star$-Brownian~\cite{Follmer85}.  Under $Q^\star$, the canonical process is a nonexplosive weak solution of the feedback SDE $dX_t=u^\star_t\,dt+dW_t$ on $[0,\Delta)$.  For every $T<\Delta$, the feedback $b(t,y):=\partial_y\log q_t(y)$ is locally Lipschitz in $y$, locally uniformly for $t\le T$, so pathwise uniqueness holds up to explosion.  The localized Yamada--Watanabe argument then yields a strong solution on $[0,T]$, consistent as $T\uparrow\Delta$. Thus $u^\star$ is progressively measurable with respect to the completed natural filtration of $W$.  Equality in \eqref{eq:gibbs-terminal-cost} therefore gives
    \[
        a\E_{Q^\star}\int_0^\Delta (u_t^\star)^2\,dt
        +
        \E_{Q^\star}\Gamma(X_\Delta)
        =
        -2a\log Z.
    \]
    This proves the reverse inequality.

    If $Z=0$, then $e^{-\Gamma(X_\Delta)/(2a)}=0$ almost surely under $\mathsf W_x$.  The Gibbs variational principle and the Girsanov entropy bound imply that no finite-energy control can have finite terminal cost, so both sides of \eqref{eq:entropic-identity} are $+\infty$.
    \end{proof}

    \begin{theorem}
    \label{thm:weighted-recursion}
    Let
    \[
        0=t_0<t_1<\cdots<t_m=1,
        \qquad
        \eta(t)=a_i>0
        \quad (t_i\le t<t_{i+1}),
    \]
    and write $\Delta_i=t_{i+1}-t_i$.  Define
    \begin{equation}
        H_{m-1}
        :=
        P_{\Delta_{m-1}}\1_{[\kappa,\infty)}
    \label{eq:hjb-terminal}
    \end{equation}
    and, for $i=m-2,\ldots,0$,
    \begin{equation}
        H_i
        :=
        P_{\Delta_i}
        \left(
            H_{i+1}^{a_{i+1}/a_i}
        \right).
    \label{eq:hjb-recursion}
    \end{equation}
    Then
    \begin{equation}
        V_\kappa(\eta)
        =
        -2a_0\log H_0(0).
    \label{eq:hjb-value}
    \end{equation}
    \end{theorem}

    \begin{proof}
    Represent the terminal constraint by the extended-valued terminal cost
    \[
        \Gamma_m(x)
        :=
        \begin{cases}
        0,&x\ge\kappa,\\
        +\infty,&x<\kappa.
        \end{cases}
    \]
    Dynamic programming at the deterministic times $t_i$ reduces the control problem to one interval at a time.  On the final interval, \cref{lem:entropic-identity} gives
    \[
        \Gamma_{m-1}
        =
        -2a_{m-1}\log H_{m-1}.
    \]
    If the continuation value at time $t_{i+1}$ is $-2a_{i+1}\log H_{i+1}$, then changing the running-cost coefficient from $a_{i+1}$ to $a_i$ gives
    \[
        \exp\!\left(
            -\frac{\Gamma_{i+1}}{2a_i}
        \right)
        =
        H_{i+1}^{a_{i+1}/a_i}.
    \]
    Applying \cref{lem:entropic-identity} on $[t_i,t_{i+1}]$ yields \eqref{eq:hjb-recursion}.  Backward iteration and evaluation at the initial state give \eqref{eq:hjb-value}.
    \end{proof}

    The recursion is monotone, which is the key fact used by the verified numerical lower bound.

    \begin{lemma}
    \label{lem:hjb-monotone-envelopes}
    Every $H_i$ in \cref{thm:weighted-recursion} takes values in $[0,1]$ and is nondecreasing.  If $U_{i+1}\ge H_{i+1}$ pointwise and $r_i=a_{i+1}/a_i$, then
    \[
        P_{\Delta_i}
        \left(
            \bigl((U_{i+1}\vee0)\wedge1\bigr)^{r_i}
        \right)
        \ge H_i.
    \]
    More generally, if an upper bound for $H_{i+1}^{r_i}$ is available on a bounded interval, extending that bound by $1$ outside the interval preserves the upper-bound direction after applying $P_{\Delta_i}$.  In particular, any rigorous bound $H_0(0)\le U_0$ implies
    \[
        V_\kappa(\eta)
        \ge
        -2a_0\log U_0.
    \]
    \end{lemma}

    \begin{proof}
    The terminal function in \eqref{eq:hjb-terminal} is nondecreasing and takes values in $[0,1]$.  The map $x\mapsto x^{r_i}$ is increasing on $[0,1]$, and the heat semigroup preserves both pointwise order and monotonicity. Backward induction proves the first two assertions.  Outside a verified interval, the trivial bound $H_{i+1}^{r_i}\le1$ has the correct direction. The final assertion follows because $-\log$ is decreasing.
    \end{proof}

    \begin{lemma}
    \label{lem:weighted-certificate}
    There are $m=68$ rational interval lengths $\Delta_i>0$ and strictly increasing rational weights $a_i>0$ satisfying
    \[
        \sum_{i=0}^{67}\Delta_i=1,
        \qquad
        \sum_{i=0}^{67}\Delta_i a_i=1,
    \]
    for which a rigorous outward-rounded evaluation of \eqref{eq:hjb-terminal}--\eqref{eq:hjb-recursion} at $\kappa=0$ proves
    \begin{equation}
        -2a_0\log H_0(0)
        >
        \frac{173639}{100000}.
    \label{eq:weighted-certificate-output}
    \end{equation}
    \end{lemma}

    The verification uses exact rational time intervals and weights together with piecewise-affine upper bounds for the functions in the recursion. A curvature estimate controls interpolation between grid points, directed rounding controls the Gaussian convolutions, and analytic tail bounds extend the inequalities to the whole real line.  The proof of soundness is given in \cref{app:lower-certificate}.

    \begin{proposition}
    \label{prop:zero-lower}
    \[
        \Estar(0)
        >
        \frac{173639}{100000}.
    \]
    \end{proposition}

    \begin{proof}
    Apply \cref{lem:weighted-lower-bound,thm:weighted-recursion,lem:weighted-certificate}.
    \end{proof}

    \paragraph{A feasible control.}
    For the upper bound, we construct a self-similar diffusion whose drift has a controlled mean and variance.  A deterministic correction then makes its second moment equal to a prescribed constant at every time while preserving the terminal inequality.  The only computer-assisted ingredient is the verification of the profile inequalities in the next lemma.

    \begin{lemma}
    \label{lem:certified-profile}
    Set
    \[
        C=\frac{787}{1000},
        \qquad
        \Lambda=\frac{10207}{10000},
        \qquad
        \lambda=\frac{5003}{2000},
    \]
    \[
        G_0=\frac{3359}{2000},
        \qquad
        R_0=\frac{30293}{200000},
        \qquad
        \mathcal E_+=\frac{243}{125}=1.944.
    \]
    There is a positive, globally $C^1$, piecewise $C^3$ function $g:\R\to(0,\infty)$ with bounded derivative and at most linear growth such that, with
    \[
        \cL f(z)
        :=
        \frac12f''(z)
        +
        \left(\frac z2+g(z)\right)f'(z),
        \qquad
        \mathfrak r(z)
        :=
        Cg(z)-g'(z)^2,
    \]
    the following properties hold.

    \begin{enumerate}[label=(\roman*),leftmargin=2.2em]
    \item For almost every $z\in\R$,
    \begin{equation}
        g(z)+z\ge0,
        \qquad
        -\Lambda g(z)\le\cL g(z)\le-g(z),
        \qquad
        \mathfrak r(z)\ge0,
        \qquad
        \cL\mathfrak r(z)\ge-\lambda\mathfrak r(z).
    \label{eq:profile-inequalities}
    \end{equation}

    \item At every interface at which $\mathfrak r'$ has a jump,
    \[
        [\mathfrak r']_\xi
        :=
        \mathfrak r'(\xi+)-\mathfrak r'(\xi-)
        \ge0.
    \]

    \item Writing
    \[
        \gamma
        :=
        g(0)
        =
        \frac{209924176232601}{125000000000000},
    \]
    one has $\gamma<G_0$ and $\mathfrak r(0)>R_0$.

    \item For $0\le s\le1$, define
    \begin{equation}
        \overline v(s)
        :=
        CG_0
        +
        \left(
            G_0^2-CG_0-\frac{R_0}{\lambda-2}
        \right)s
        +
        \frac{R_0}{\lambda-2}s^{\lambda-1}
        -
        \gamma^2s^{2\Lambda-1}.
    \label{eq:vbar}
    \end{equation}
    Then
    \[
        \overline v(s)<\mathcal E_+
        \qquad (0\le s\le1).
    \]

    \item One has
    \begin{equation}
        \int_0^1
        \sqrt{\mathcal E_+-\overline v(s)}\,ds
        >
        \frac{2G_0}{3}.
    \label{eq:barrier-integral}
    \end{equation}
    \end{enumerate}
    \end{lemma}

    The profile is represented by a rational spline on a compact interval and analytic expressions on the two tails.  Exact Bernstein-form inequalities verify the spline pieces, analytic estimates handle the tails and the interface terms, and a rational lower sum verifies \eqref{eq:barrier-integral}.  These checks are detailed in \cref{app:upper-certificate}.

    \begin{proposition}
    \label{prop:zero-upper}
    \[
        \Estar(0)\le\frac{243}{125}=1.944.
    \]
    \end{proposition}

    \begin{proof}
    Let $Z$ solve
    \begin{equation}
        dZ_\tau
        =
        dW_\tau
        +
        \left(
            \frac{Z_\tau}{2}+g(Z_\tau)
        \right)d\tau,
        \qquad
        Z_0=0.
    \label{eq:self-similar-sde}
    \end{equation}
    The drift is globally Lipschitz and has linear growth, so the solution is strong and nonexplosive.  Set
    \[
        s=e^{-\tau},
        \qquad
        t=1-s,
        \qquad
        X_t=\sqrt s\,Z_\tau,
    \]
    and define
    \[
        B_t
        :=
        \int_0^{-\log(1-t)}
        e^{-r/2}\,dW_r.
    \]
    The quadratic variation of $B$ is $t$, hence $B$ is Brownian motion.  The time change is deterministic and invertible, so the completed natural filtration of $B$ at time $t$ is the completed natural filtration of $W$ at time $-\log(1-t)$.  It\^o's formula gives
    \begin{equation}
        dX_t
        =
        dB_t+h_t\,dt,
        \qquad
        h_t
        =
        \frac1{\sqrt{1-t}}
        g\left(
            \frac{X_t}{\sqrt{1-t}}
        \right).
    \label{eq:self-similar-feedback}
    \end{equation}
    Thus $h$ is progressively measurable with respect to the Brownian filtration.

    Define
    \[
        \mu(\tau):=\E g(Z_\tau),
        \qquad
        S(\tau):=\E g(Z_\tau)^2,
        \qquad
        \varrho(\tau):=\E\mathfrak r(Z_\tau).
    \]
    After localization, the inequalities in \eqref{eq:profile-inequalities} imply that $e^\tau g(Z_\tau)$ is a supermartingale and $e^{\Lambda\tau}g(Z_\tau)$ is a submartingale.  Therefore
    \begin{equation}
        \gamma e^{-\Lambda\tau}
        \le
        \mu(\tau)
        \le
        \gamma e^{-\tau}
        <
        G_0e^{-\tau}.
    \label{eq:mu-bounds}
    \end{equation}
    For $\mathfrak r$, the generalized It\^o--Tanaka formula contains the interface local-time terms $\frac12[\mathfrak r']_\xi L_\tau^\xi(Z)$, which are nonnegative by \cref{lem:certified-profile}.  Hence
    \begin{equation}
        \varrho(\tau)
        \ge
        \mathfrak r(0)e^{-\lambda\tau}
        >
        R_0e^{-\lambda\tau}.
    \label{eq:varrho-bound}
    \end{equation}

    Moreover,
    \[
        \cL(g^2)
        =
        2g\cL g+g'^2
        \le
        -2g^2+Cg-\mathfrak r.
    \]
    Combining this inequality with \eqref{eq:mu-bounds}--\eqref{eq:varrho-bound} gives
    \begin{equation}
        S'(\tau)
        \le
        -2S(\tau)
        +
        CG_0e^{-\tau}
        -
        R_0e^{-\lambda\tau}.
    \label{eq:S-differential}
    \end{equation}
    Since $S(0)=\gamma^2\le G_0^2$, integration yields
    \begin{equation}
    \begin{split}
        S(\tau)
        \le{}&
        G_0^2e^{-2\tau}
        +
        CG_0(e^{-\tau}-e^{-2\tau})\\
        &-
        \frac{R_0}{\lambda-2}
        \left(
            e^{-2\tau}-e^{-\lambda\tau}
        \right).
    \end{split}
    \label{eq:S-bound}
    \end{equation}

    Let
    \[
        m(t):=\E h_t,
        \qquad
        v(t):=\operatorname{Var}(h_t).
    \]
    Since $s=e^{-\tau}$,
    \[
        m(t)=s^{-1/2}\mu(\tau),
        \qquad
        v(t)
        =
        s^{-1}\bigl(S(\tau)-\mu(\tau)^2\bigr).
    \]
    Using the lower bound on $\mu$ in \eqref{eq:mu-bounds} and the upper bound \eqref{eq:S-bound} gives
    \begin{equation}
        v(t)\le\overline v(1-t),
    \label{eq:v-vbar}
    \end{equation}
    while the upper bound on $\mu$ gives
    \begin{equation}
        \int_0^1m(t)\,dt
        =
        \int_0^\infty e^{-\tau/2}\mu(\tau)\,d\tau
        <
        \frac{2G_0}{3}.
    \label{eq:mean-integral}
    \end{equation}
    In particular, $\sup_{t<1}\E h_t^2<\infty$.  Hence $\int_0^1|h_t|\,dt<\infty$ almost surely and $X_t$ has an almost sure terminal limit $X_1$.

    The first inequality in \eqref{eq:profile-inequalities} gives
    \[
        \frac1{\sqrt{1-t}}
        g\left(
            \frac{x}{\sqrt{1-t}}
        \right)
        \ge
        -\frac{x}{1-t}.
    \]
    On each interval $[0,T]$ with $T<1$, comparison with
    \[
        dK_t
        =
        dB_t-\frac{K_t}{1-t}\,dt,
        \qquad K_0=0,
    \]
    therefore gives $X_t\ge K_t$.  The process $K$ is a Brownian bridge pinned at zero, so $K_t\to0$ almost surely as $t\uparrow1$.  Consequently,
    \begin{equation}
        X_1\ge0
        \qquad\text{a.s.}
    \label{eq:X1-positive}
    \end{equation}

    It remains to control the second moment of the drift uniformly in time. Define the deterministic function
    \begin{equation}
        \ell(t)
        :=
        \sqrt{\mathcal E_+-v(t)}-m(t),
        \qquad
        u_t:=h_t+\ell(t).
    \label{eq:flattening-control}
    \end{equation}
    This is well-defined by \eqref{eq:v-vbar} and \cref{lem:certified-profile}.  Since $\ell$ is deterministic,
    \begin{equation}
        \E u_t^2
        =
        v(t)+(m(t)+\ell(t))^2
        =
        \mathcal E_+
    \label{eq:flat-energy}
    \end{equation}
    for almost every $t$.  Furthermore,
    \[
    \begin{split}
        \int_0^1\ell(t)\,dt
        &=
        \int_0^1\sqrt{\mathcal E_+-v(t)}\,dt
        -
        \int_0^1m(t)\,dt\\
        &\ge
        \int_0^1
        \sqrt{\mathcal E_+-\overline v(1-t)}\,dt
        -
        \frac{2G_0}{3}
        >0
    \end{split}
    \]
    by \eqref{eq:barrier-integral} and \eqref{eq:mean-integral}.  Combining this with \eqref{eq:X1-positive},
    \[
        B_1+\int_0^1u_t\,dt
        =
        X_1+\int_0^1\ell(t)\,dt
        >0
        \qquad\text{a.s.}
    \]
    Thus $u$ is feasible at zero margin and $\esssup_t\E u_t^2=\mathcal E_+=243/125$.
    \end{proof}

    \begin{proof}[Proof of \cref{thm:zero}]
    By \cref{prop:zero-lower,prop:zero-upper},
    \[
        1.73639<\Estar(0)\le1.944.
    \]
    Applying \cref{thm:main} at $\kappa=0$ and using that $x\mapsto2/(\pi x)$ is decreasing gives
    \[
        \frac{2}{1.944\pi}
        \le
        \aon(0)
        <
        \frac{2}{1.73639\pi}.
    \]
    Since $1.944=243/125$ and $1.73639=173639/100000$, these endpoints are
    \[
        \frac{250}{243\pi}
        =
        0.3274793\ldots,
        \qquad
        \frac{200000}{173639\pi}
        =
        0.3666340\ldots,
    \]
    which is \eqref{eq:zero-bracket}.
    \end{proof}

    \begin{remark}[Numerics for the weighted relaxation]
    \label{rem:weighted-numerics}
    Numerical optimization of the weighted $L^2$ relaxation suggests that its best zero-margin value is approximately $1.7373$.  For example, within the family
    \[
        \eta_a(t)
        =
        a+\frac{1-a}{2\sqrt{1-t}},
    \]
    mesh refinement places the optimum near $a=0.340$ and extrapolates the relaxed value to about $1.7372$--$1.7374$.  These computations concern the lower-bound relaxation only; whether $\sup_\eta V_0(\eta)=\Estar(0)$ is an open question.
    \end{remark}

    \subsection{Large positive margin}
    \label{sec:large-positive}

    We next consider $\kappa\to+\infty$.  Only elementary properties of $\Estar$ are needed.

    \begin{lemma}
    \label{lem:estar-lipschitz}
    If $\kappa_1\le\kappa_2$, then
    \[
        \Estar(\kappa_1)\le\Estar(\kappa_2),
    \]
    and
    \[
        \sqrt{\Estar(\kappa_2)}
        \le
        \sqrt{\Estar(\kappa_1)}
        +
        \kappa_2-\kappa_1.
    \]
    \end{lemma}

    \begin{proof}
    The first assertion follows from inclusion of the terminal constraints. For the second, add the deterministic drift $\kappa_2-\kappa_1$ to a control feasible at level $\kappa_1$ and use Minkowski's inequality in $L^2$ at each time.  Taking the infimum over controls arbitrarily close to optimal gives the claim.
    \end{proof}

    \begin{proposition}
    \label{prop:large-margin}
    As $\kappa\to+\infty$,
    \[
        \Estar(\kappa)
        =
        \kappa^2\bigl(1+O(\kappa^{-1})\bigr),
    \]
    and consequently
    \[
        \aon(\kappa)
        =
        \frac{2}{\pi\kappa^2}
        \bigl(1+O(\kappa^{-1})\bigr).
    \]
    \end{proposition}

    \begin{proof}
    By \cref{lem:elementary-E-bounds},
    \[
        \Estar(\kappa)
        \ge
        \E(\kappa-Z)_+^2
        =
        (\kappa^2+1)\Ph(\kappa)
        +
        \kappa\ph(\kappa),
        \qquad Z\sim N(0,1).
    \]
    For $\kappa>0$, Mills' inequality $1-\Ph(\kappa)\le\ph(\kappa)/\kappa$ gives
    \[
    \begin{aligned}
        \E(\kappa-Z)_+^2-\kappa^2
        &=
        \Ph(\kappa)
        +
        \kappa\ph(\kappa)
        -
        \kappa^2(1-\Ph(\kappa))\\
        &\ge
        \Ph(\kappa)>0.
    \end{aligned}
    \]
    Thus $\Estar(\kappa)\ge\kappa^2$.  On the other hand, \cref{lem:estar-lipschitz} with $\kappa_1=0$ gives
    \[
        \Estar(\kappa)
        \le
        \bigl(
            \kappa+\sqrt{\Estar(0)}
        \bigr)^2
        =
        \kappa^2\bigl(1+O(\kappa^{-1})\bigr).
    \]
    The assertion for $\aon(\kappa)$ follows from \cref{thm:main}.
    \end{proof}

    \begin{remark}[Relation to the offline upper bound]
    For $\kappa\ge0$, \cref{lem:elementary-E-bounds,thm:main} imply
    \[
        \aon(\kappa)
        \le
        \frac{2}{
            \pi\,\E(\kappa-Z)_+^2
        },
        \qquad Z\sim N(0,1).
    \]
    The right-hand side is Stojnic's rigorous upper bound for the offline storage capacity~\cite{Stojnic13}.  In particular, the online threshold and the known offline upper bound have the same first-order asymptotic $2/(\pi\kappa^2)$ as $\kappa\to+\infty$.
    \end{remark}

    \subsection{Large negative margin}
    \label{sec:large-negative}

    Write
    \[
        K:=-\kappa,
        \qquad
        q_K:=\Ph(-K),
        \qquad
        K\to\infty,
    \]
    and define
    \begin{equation}
        I_0
        :=
        \int_0^\infty
        \sqrt{
            2\Ph\!\left(-\sqrt{s/2}\right)
        }\,ds.
    \label{eq:I0}
    \end{equation}

    \begin{theorem}[Large negative margin]
    \label{thm:large-negative}
    As $K\to\infty$,
    \begin{equation}
        \left(I_0^{-2}+o(1)\right)K^2q_K
        \le
        \Estar(-K)
        \le
        (4+o(1))K^2q_K,
    \label{eq:large-negative-E}
    \end{equation}
    and consequently
    \begin{equation}
        \frac{1/(2\pi)+o(1)}{K^2q_K}
        \le
        \aon(-K)
        \le
        \frac{2I_0^2/\pi+o(1)}{K^2q_K}.
    \label{eq:large-negative-alpha}
    \end{equation}
    In particular, $\aon(\kappa)=\Theta(1/(\kappa^2\Ph(\kappa)))$ as $\kappa\to-\infty$.
    \end{theorem}

    We prove \cref{thm:large-negative} in two parts.  The upper bound is obtained by prescribing a compactly supported terminal distribution through a F\"ollmer drift.  For the lower bound, the relevant time scale is $1-t\asymp K^{-2}$; the required Gaussian estimate is isolated in \cref{lem:large-negative-bivariate}.

    \begin{proposition}
    \label{prop:large-negative-upper}
    As $K\to\infty$,
    \begin{equation}
        \Estar(-K)
        \le
        (4+o(1))K^2q_K.
    \label{eq:large-negative-upper}
    \end{equation}
    \end{proposition}

    \begin{proof}
    Let $\gamma=N(0,1)$ and define
    \[
    H_K(x):=
    \begin{cases}
    0,&x\le-K,\\
    1-e^{-K(x+K)},&-K<x\le0,\\
    c_K,&0<x\le K,\\
    c_K(K+1-x),&K<x<K+1,\\
    0,&x\ge K+1,
    \end{cases}
    \qquad
    c_K:=1-e^{-K^2}.
    \]
    Set
    \[
        f_K(x)
        :=
        \frac{H_K(x)^2}{
            \int H_K^2\,d\gamma
        }.
    \]
    The probability measure $f_K\,d\gamma$ is supported on $[-K,K+1]$. By \cref{lem:localized-follmer}, the associated F\"ollmer process has terminal law $f_K\,d\gamma$.  On every interval $[0,T]$ with $T<1$ the drift is smooth in the state, so the corresponding SDE has a unique strong solution; hence the drift is adapted to the completed natural filtration of the driving Brownian motion.  It is therefore feasible at level $-K$.  Its second moment at every time is at most the squared $L^2(f_K\,d\gamma)$ norm of the score
    \[
        \int s_{f_K}^2 f_K\,d\gamma
        =
        4
        \frac{
            \int(H_K')^2\,d\gamma
        }{
            \int H_K^2\,d\gamma
        }.
    \]
    For each fixed $x$, $H_K(x)\to1$, and $0\le H_K\le1$.  Hence
    \begin{equation}
        \int H_K^2\,d\gamma=1-o(1).
    \label{eq:large-negative-denominator}
    \end{equation}
    The contribution of $[-K,0]$ to the derivative energy is
    \begin{align}
        A_K
        &:=
        K^2\int_{-K}^0
        e^{-2K(x+K)}\ph(x)\,dx
        \notag\\
        &=
        K^2\ph(K)
        \int_0^K e^{-Ky-y^2/2}\,dy
        \notag\\
        &=
        K\ph(K)
        \int_0^{K^2}
        e^{-z-z^2/(2K^2)}\,dz
        \notag\\
        &=
        (1+o(1))K\ph(K).
    \label{eq:large-negative-lower-layer}
    \end{align}
    The interval $[K,K+1]$ contributes
    \[
        c_K^2\gamma([K,K+1])
        =
        O(q_K)
        =
        o(K\ph(K)).
    \]
    Therefore
    \[
        \int(H_K')^2\,d\gamma
        =
        (1+o(1))K\ph(K).
    \]
    Using Mills' ratio $\ph(K)\sim Kq_K$ proves \eqref{eq:large-negative-upper}.
    \end{proof}

    The lower bound depends on the following asymptotic for two highly correlated Gaussian variables.

    \begin{lemma}[Bivariate Gaussian tail near correlation one]
    \label{lem:large-negative-bivariate}
    For $K>0$ and $0\le t<1$, let
    \[
        p_K(t)
        :=
        \Pp(X<-K,Y<-K),
    \]
    where $(X,Y)$ is a centered bivariate normal vector with unit variances and correlation $t$.  Then, for every fixed $s\ge0$,
    \begin{equation}
        \frac{
            p_K(1-s/K^2)
        }{
            q_K
        }
        \longrightarrow
        2\Ph\!\left(-\sqrt{s/2}\right)
    \label{eq:large-negative-bivariate-limit}
    \end{equation}
    as $K\to\infty$.  Moreover, there is a universal constant $C$ such that, for all sufficiently large $K$ and all $0\le s\le K^2$,
    \begin{equation}
        \frac{
            p_K(1-s/K^2)
        }{
            q_K
        }
        \le
        C e^{-s/4}.
    \label{eq:large-negative-bivariate-envelope}
    \end{equation}
    Consequently,
    \begin{equation}
        \int_0^1\sqrt{p_K(t)}\,dt
        =
        \frac{\sqrt{q_K}}{K^2}
        \left(
            I_0+o(1)
        \right),
    \label{eq:large-negative-J}
    \end{equation}
    where
    \[
        I_0
        :=
        \int_0^\infty
        \sqrt{
            2\Ph\!\left(-\sqrt{s/2}\right)
        }\,ds.
    \]
    \end{lemma}

    \begin{proof}
    Fix $s\ge0$ and put $t=1-s/K^2$.  Set
    \[
        U:=\frac{X+Y}{2},
        \qquad
        V:=\frac{X-Y}{2}.
    \]
    Then $U$ and $V$ are independent and
    \[
        \operatorname{Var}(U)
        =
        1-\frac{s}{2K^2},
        \qquad
        \operatorname{Var}(V)
        =
        \frac{s}{2K^2},
    \]
    while
    \[
        \{X<-K,Y<-K\}
        =
        \{U<-K-|V|\}.
    \]
    If $Z\sim N(0,1)$, then
    \[
        |V|
        \stackrel{d}{=}
        \frac1K\sqrt{\frac s2}\,|Z|.
    \]
    Conditionally on $Z$, the relevant Gaussian threshold is
    \[
        \frac{K+|V|}{\sqrt{1-s/(2K^2)}}
        =K+\frac{c_K(Z)}{K},
    \]
    where
    \begin{equation}
        c_K(Z)
        :=K\left(
            \frac{K+|V|}{\sqrt{1-s/(2K^2)}}-K
        \right)
        \longrightarrow
        \frac{s}{4}+\sqrt{\frac s2}|Z|.
    \label{eq:large-negative-shift}
    \end{equation}
    The standard tail ratio
    \[
        \frac{\Ph(-(K+c/K))}{\Ph(-K)}
        \longrightarrow e^{-c}
    \]
    therefore gives the conditional pointwise limit.  To pass the limit through the expectation over $Z$, use the two-sided Mills bounds: for $y\ge0$ and all sufficiently large $K$,
    \[
        \frac{\Ph(-(K+y/K))}{\Ph(-K)}\le C e^{-y}.
    \]
    Since $c_K(Z)\ge\sqrt{s/2}|Z|$, this is an integrable dominating function under the standard Gaussian law.  Dominated convergence yields
    \[
    \begin{split}
        \frac{p_K(1-s/K^2)}{q_K}
        &\longrightarrow
        e^{-s/4}
        \E
        \exp\!\left(
            -\sqrt{\frac s2}|Z|
        \right)\\
        &=
        2\Ph\!\left(-\sqrt{s/2}\right),
    \end{split}
    \]
    proving \eqref{eq:large-negative-bivariate-limit}.

    For the uniform bound, observe that
    \[
        \{X<-K,Y<-K\}
        \subseteq
        \left\{
            \frac{X+Y}{\sqrt{2(1+t)}}
            <
            -K\sqrt{\frac{2}{1+t}}
        \right\}.
    \]
    The random variable on the left of the inequality is standard normal. Standard Mills bounds therefore give, uniformly for $0\le t<1$ and all large $K$,
    \[
        \frac{p_K(t)}{q_K}
        \le
        C
        \exp\!\left(
            -\frac{K^2(1-t)}{4}
        \right).
    \]
    Taking $t=1-s/K^2$ yields \eqref{eq:large-negative-bivariate-envelope}.

    Finally, changing variables $s=K^2(1-t)$ gives
    \[
        \int_0^1\sqrt{p_K(t)}\,dt
        =
        \frac{\sqrt{q_K}}{K^2}
        \int_0^{K^2}
        \sqrt{
            \frac{p_K(1-s/K^2)}{q_K}
        }\,ds.
    \]
    The pointwise limit is $\sqrt{2\Ph(-\sqrt{s/2})}$, while \eqref{eq:large-negative-bivariate-envelope} provides the integrable dominating function $C^{1/2}e^{-s/8}$.  Dominated convergence proves \eqref{eq:large-negative-J}.
    \end{proof}

    \begin{proposition}
    \label{prop:large-negative-lower}
    As $K\to\infty$,
    \begin{equation}
        \Estar(-K)
        \ge
        \left(
            I_0^{-2}+o(1)
        \right)
        K^2q_K.
    \label{eq:large-negative-lower}
    \end{equation}
    \end{proposition}

    \begin{proof}
    Let $u$ be feasible at level $-K$.  Replacing $u$ by $u_+:=\max\{u,0\}$ preserves feasibility and does not increase $\E u_t^2$, so we may assume $u_t\ge0$.  Put
    \[
        A:=\int_0^1u_t\,dt,
        \qquad
        I_K:=\1_{\{B_1<-K\}},
        \qquad
        M_t^K:=\E[I_K\mid\cF_t^B].
    \]
    On $\{B_1<-K\}$, feasibility gives
    \[
        A\ge -K-B_1.
    \]
    Consequently, if
    \[
        m_K
        :=
        \E(-K-B_1)_+,
    \]
    then
    \begin{align}
        m_K
        &\le
        \E[I_KA]
        =
        \int_0^1\E[M_t^Ku_t]\,dt
        \notag\\
        &\le
        \sqrt E
        \int_0^1\norm{M_t^K}_{L^2}\,dt,
    \label{eq:large-negative-test}
    \end{align}
    where
    \[
        E
        :=
        \esssup_{0<t<1}\E u_t^2.
    \]
    Hence
    \begin{equation}
        E
        \ge
        \left(
            \frac{m_K}{J_K}
        \right)^2,
        \qquad
        J_K
        :=
        \int_0^1\norm{M_t^K}_{L^2}\,dt.
    \label{eq:large-negative-E-lower-pre}
    \end{equation}

    Mills' ratio gives
    \begin{equation}
        m_K
        =
        \ph(K)-Kq_K
        =
        (1+o(1))
        \frac{q_K}{K}.
    \label{eq:large-negative-m}
    \end{equation}
    If $X,Y$ are standard Gaussians with correlation $t$, then
    \begin{equation}
        \norm{M_t^K}_{L^2}^2
        =
        \Pp(X<-K,Y<-K)
        =
        p_K(t).
    \label{eq:large-negative-bivariate}
    \end{equation}
    Thus \cref{lem:large-negative-bivariate} gives
    \[
        J_K
        =
        \frac{\sqrt{q_K}}{K^2}
        \left(
            I_0+o(1)
        \right).
    \]
    Substituting this and \eqref{eq:large-negative-m} into \eqref{eq:large-negative-E-lower-pre} yields
    \[
        E
        \ge
        \left(
            I_0^{-2}+o(1)
        \right)
        K^2q_K.
    \]
    Taking the infimum over feasible $u$ proves \eqref{eq:large-negative-lower}.
    \end{proof}

    \begin{proof}[Proof of \cref{thm:large-negative}]
    Combining \cref{prop:large-negative-upper,prop:large-negative-lower} gives
    \[
        \left(
            I_0^{-2}+o(1)
        \right)
        K^2q_K
        \le
        \Estar(-K)
        \le
        (4+o(1))K^2q_K.
    \]
    The numerical value
    \[
        I_0=5.06296006456\ldots
    \]
    is obtained by one-dimensional numerical integration and is used only for display.  Applying \cref{thm:main} gives
    \[
        \frac{1/(2\pi)+o(1)}{K^2q_K}
        \le
        \aon(-K)
        \le
        \frac{2I_0^2/\pi+o(1)}{K^2q_K},
    \]
    and in particular
    \[
        \aon(\kappa)
        =
        \Theta\!\left(
            \frac{1}{
                \kappa^2\Ph(\kappa)
            }
        \right)
        \qquad
        (\kappa\to-\infty).
    \]
    \end{proof}
\section*{AI Usage Disclosure}

OpenAI's GPT-5.6 Sol was used for exploratory discussions during the
development of the proofs and to assist with drafting and revising
the manuscript. The authors take full responsibility for the paper's
contents and correctness.

\appendix
    \crefalias{section}{appendix}

    \section{Proof of the Brownian limit}
    \label{app:mean-field}

    This appendix proves \cref{prop:mean-field-limit}.  The argument follows the Doob-decomposition and compactness method of \cite[Section~3]{FiedlerJacksonLackerNilesWeed26}, specialized to the Gaussian perceptron and to a uniformly chosen row.  We use the normalization $g_k/\sqrt N$ throughout, so the limiting process already lives on the time interval $[0,1]$.

    By the seed-fixing argument in \cref{sec:preliminaries}, it is enough to consider deterministic online algorithms.  Put
    \[
        Z_k:=\frac{g_k}{\sqrt N},
    \]
    and let $\cH_{k-1}$ be the sigma-field generated by the columns revealed before $g_k$.  Define
    \begin{equation}
        b_k:=\E[\sigma_k Z_k\mid\cH_{k-1}],
        \qquad
        d_k:=\sigma_k Z_k-b_k.
    \label{eq:appendix-doob}
    \end{equation}
    By \cref{lem:one-column-drift-bound},
    \begin{equation}
        \norm{b_k}_2^2\le\frac{2}{\pi N}
        \qquad\text{a.s.}
    \label{eq:appendix-drift-ball}
    \end{equation}
    For $1\le i\le M$, linearly interpolate the paths from the mesh values
    \begin{equation}
        A_i^N(k/N):=\sum_{j=1}^k b_{j,i},
        \qquad
        W_i^N(k/N):=\sum_{j=1}^k d_{j,i},
        \qquad 0\le k\le N.
    \label{eq:appendix-paths}
    \end{equation}
    Let $I_N$ be uniform on $\{1,\ldots,M\}$ and independent of the matrix, and write
    \[
        A^N:=A_{I_N}^N,
        \qquad
        W^N:=W_{I_N}^N.
    \]

    \subsection{Tightness and the drift bound}

    On the interval $((k-1)/N,k/N)$,
    \[
        \frac1M\sum_{i=1}^M|\dot A_i^N(t)|^2
        =\frac{N^2}{M}\norm{b_k}_2^2
        \le\frac{2N}{\pi M}.
    \]
    Hence, for every nonnegative continuous $\psi$ on $[0,1]$,
    \begin{equation}
        \E\int_0^1\psi(t)|\dot A^N(t)|^2\,dt
        \le
        \frac{2N}{\pi M}\int_0^1\psi(t)\,dt.
    \label{eq:appendix-weighted-energy-N}
    \end{equation}
    In particular, the $H^1$ norms of $A^N$ are bounded in probability.  Since $A^N(0)=0$, the compact embedding of bounded subsets of $H^1([0,1])$ into $C([0,1])$ implies tightness of the laws of $A^N$.

    The martingale increments satisfy
    \begin{equation}
        \E[d_k\mid\cH_{k-1}]=0,
        \qquad
        \operatorname{Cov}(d_k\mid\cH_{k-1})
        =\frac1N I_M-b_kb_k^{\mathsf T}
        \preceq\frac1N I_M.
    \label{eq:appendix-covariance}
    \end{equation}
    Moreover, for a universal constant $C$,
    \begin{equation}
        \E[d_{k,i}^4\mid\cH_{k-1}]
        \le\frac{C}{N^2},
        \qquad 1\le i\le M.
    \label{eq:appendix-fourth-moment}
    \end{equation}
    Indeed, $\sigma_k^2=1$, $Z_{k,i}\sim N(0,1/N)$ conditionally on $\cH_{k-1}$ before the current column is revealed, and $|b_{k,i}|^2\le\norm{b_k}_2^2=O(N^{-1})$. The standard martingale tightness criterion applied after adjoining $I_N$ to the filtration, using \eqref{eq:appendix-covariance} and \eqref{eq:appendix-fourth-moment}, gives tightness of the laws of $W^N$. Thus the joint laws of $(A^N,W^N)$ are tight in $C([0,1])^2$.

    Pass to a subsequence such that
    \begin{equation}
        (A^N,W^N)\Rightarrow(A,B)
        \qquad\text{in }C([0,1])^2.
    \label{eq:appendix-joint-limit}
    \end{equation}
    For nonnegative continuous $\psi$, define
    \[
    I_\psi(a):=
    \begin{cases}
    \displaystyle\int_0^1\psi(t)|\dot a(t)|^2\,dt,
    & a(0)=0\text{ and }a\text{ is absolutely continuous},\\
    +\infty,&\text{otherwise}.
    \end{cases}
    \]
    This functional is lower semicontinuous on $C([0,1])$.  By \eqref{eq:appendix-weighted-energy-N}, the Portmanteau theorem, and $N/M\to1/\alpha$,
    \begin{equation}
        \E I_\psi(A)
        \le
        \frac{2}{\pi\alpha}\int_0^1\psi(t)\,dt.
    \label{eq:appendix-weighted-energy}
    \end{equation}
    Taking $\psi\equiv1$ shows that $A$ is absolutely continuous almost surely. Write
    \[
        A_t=\int_0^t a_s\,ds.
    \]
    An adapted absolutely continuous process admits a progressively measurable version of its derivative; for example one may take a pointwise limit of left difference quotients on the set where the derivative exists.  From \eqref{eq:appendix-weighted-energy}, first for a countable dense family of nonnegative continuous $\psi$ and then by approximation,
    \begin{equation}
        \E a_t^2\le\frac{2}{\pi\alpha}
        \qquad\text{for a.e. }t\in[0,1].
    \label{eq:appendix-pointwise-energy}
    \end{equation}

    \subsection{Identification of the Brownian motion}

    Let $\cG_t$ be the completed filtration generated by $(A_s,B_s)_{0\le s\le t}$.  We show that $B$ is a $\cG$-Brownian motion by passing the Brownian martingale problem to the limit.

    Fix rational $0\le s<t\le1$, a function $\varphi\in C_c^3(\R)$, and a bounded Lipschitz cylinder functional $H$ of the two paths up to time $s$. Put
    \[
        k_s:=\lfloor Ns\rfloor,
        \qquad
        k_t:=\lfloor Nt\rfloor,
    \]
    and let $H_N$ be $H$ evaluated on the two linearly interpolated paths stopped at $k_s/N$.  After adjoining $I_N$ to the discrete filtration, $H_N$ is measurable before all increments with index $k>k_s$.  Replacing the stop at $s$ by the stop at $k_s/N$ creates an $o_{L^1}(1)$ error.  Indeed, the cylinder functional $H$ is Lipschitz, while over the single boundary mesh interval the predictable part has expected increment $O(N^{-1/2})$ by Cauchy--Schwarz and \eqref{eq:appendix-weighted-energy-N}, and the martingale part has expected increment $O(N^{-1/2})$ by \eqref{eq:appendix-covariance}.

    Write
    \[
        w_k:=W^N(k/N),
        \qquad
        \Delta w_k=d_{k,I_N}.
    \]
    Taylor's formula gives
    \[
        \varphi(w_k)-\varphi(w_{k-1})
        =\varphi'(w_{k-1})\Delta w_k
         +\frac12\varphi''(w_{k-1})(\Delta w_k)^2
         +R_k,
    \]
    with $|R_k|\le C|\Delta w_k|^3$.  Conditional expectation and \eqref{eq:appendix-covariance} yield
    \begin{align}
    &\E\left[
    H_N\left(
    \varphi(w_{k_t})-\varphi(w_{k_s})
    -\frac1{2N}\sum_{k=k_s+1}^{k_t}\varphi''(w_{k-1})
    \right)
    \right]
    \notag\\
    &\qquad
    =-\frac12\E\left[
    H_N\sum_{k=k_s+1}^{k_t}
    \varphi''(w_{k-1})b_{k,I_N}^2
    \right]
    +O\!\left(\sum_{k=k_s+1}^{k_t}\E|d_{k,I_N}|^3\right).
    \label{eq:appendix-martingale-problem-discrete}
    \end{align}
    The first term on the right is $O(M^{-1})$, because
    \[
        \E\sum_{k=1}^N b_{k,I_N}^2
        =\frac1M\sum_{k=1}^N\E\norm{b_k}_2^2
        \le\frac{2}{\pi M}.
    \]
    By \eqref{eq:appendix-fourth-moment}, $\E|d_{k,I_N}|^3=O(N^{-3/2})$, so the second term in \eqref{eq:appendix-martingale-problem-discrete} is $O(N^{-1/2})$. Finally, replacing $w_{k_s},w_{k_t}$ by $W_s^N,W_t^N$ and replacing the mesh interval $[k_s/N,k_t/N]$ by $[s,t]$ changes the left-hand side by $o(1)$ in $L^1$.  The endpoint terms use the same one-mesh increment bounds as above, and the $N^{-1}$ Riemann sum differs from $\int_s^t\varphi''(W^N_r)\,dr$ by $o(1)$ because $\varphi''$ is Lipschitz, the paths are linearly interpolated, and $\E|d_{k,I_N}|=O(N^{-1/2})$.

    After these replacements, the integrand is a bounded continuous functional of the pair of paths in the uniform topology.  Passing to the weak limit in \eqref{eq:appendix-joint-limit} gives
    \[
        \E\left[
    H\left(
    \varphi(B_t)-\varphi(B_s)
    -\frac12\int_s^t\varphi''(B_r)\,dr
    \right)
    \right]=0.
    \]
    To make the filtration statement simultaneous, choose rational $s<t$, a countable $C^2$-dense family of compactly supported $C^3$ test functions, and a countable algebra of bounded Lipschitz cylinder functions with rational observation times.  Intersecting the corresponding probability-one events gives the martingale identities for this determining class.  Approximation in the test function, a monotone-class argument for bounded $\sigma((A_r,B_r):r\le s)$-measurable variables, and continuity in $s,t$ extend the identity to the natural filtration generated by $(A,B)$. Completing this filtration does not change conditional expectations.  Thus $B$ solves the Brownian martingale problem with respect to the completed filtration $\cG$, and hence is a $\cG$-Brownian motion.

    \subsection{The terminal constraint}

    At terminal time,
    \begin{equation}
        A^N(1)+W^N(1)=X_{N,I_N}.
    \label{eq:appendix-terminal-relation}
    \end{equation}
    By assumption,
    \[
        \Pp(X_{N,I_N}<\kappa)
        \le
        \Pp\left(\min_{1\le i\le M}X_{N,i}<\kappa\right)
        \longrightarrow0.
    \]
    Together with \eqref{eq:appendix-joint-limit}, this implies
    \[
        A_1+B_1\ge\kappa
        \qquad\text{a.s.}
    \]
    Combining this terminal constraint with \eqref{eq:appendix-pointwise-energy} and the Brownian identification proves \cref{prop:mean-field-limit}.
    \section{Gaussian coupling for prescribed drifts}
\label{app:coupling}

This appendix proves \cref{prop:conditional-coupling}.  The construction is the Gaussian coupling of \cite[Lemma~4.7]{FiedlerJacksonLackerNilesWeed26}, written with $M$ coordinates and $N$ arrivals and conditioned on an initial sigma-field. The two points needed in \cref{sec:achievability} are that the Gaussian steering map is invertible, so the two discrete-time filtrations agree, and that the coupling error has a conditional $\ell^\infty$ bound uniform in $M$ and $N$.

Throughout, let
\[
    t_k:=\frac{k}{N},
    \qquad
    \Delta B_{k+1}:=B_{t_{k+1}}-B_{t_k}.
\]
When the coupling is started after an initial block, all variables determined before the starting time are included in the initial sigma-field.  Brownian independent increments then reduce the argument to the notation below.

\subsection{The discrete filtrations}

The invertibility in \cref{lem:one-column-steering} is needed here for more than reconstructing the Gaussian input: it shows that the Brownian mesh and the Gaussian inputs contain exactly the same information at every discrete time.

\begin{lemma}
\label{lem:coupling-filtration}
Let $\cH_0$ be a sigma-field independent of $\Delta B_1,\ldots,\Delta B_N$, and suppose that, for each $0\le k<N$,
\[
    b_k
    \in
    L^0\!\left(
        \cH_0\vee\sigma(B_{t_0},\ldots,B_{t_k});
        \R^M
    \right),
    \qquad
    \norm{b_k}_2\le\frac{\cgauss}{\sqrt N}.
\]
Construct $(Z_{k+1},\varepsilon_{k+1})$ from $\Delta B_{k+1}$ by the inverse map in \cref{lem:one-column-steering}, using the drift $b_k$.  Then $Z_1,\ldots,Z_N$ are independent $N(0,I_M/N)$ vectors, and for every $0\le k\le N$,
\begin{equation}
    \cH_0\vee\sigma(Z_1,\ldots,Z_k)
    =
    \cH_0\vee\sigma(B_{t_0},\ldots,B_{t_k})
\label{eq:coupling-filtration-equality}
\end{equation}
after completion.
\end{lemma}

\begin{proof}
The claim is proved by induction.  It is immediate at $k=0$.  Suppose that \eqref{eq:coupling-filtration-equality} holds at time $k$.  Then $b_k$ is measurable with respect to either side of that equality.

Conditionally on the past, $\Delta B_{k+1}\sim N(0,I_M/N)$.  Applying the inverse of the map in \cref{lem:one-column-steering} to $(b_k,\Delta B_{k+1})$ produces a vector $Z_{k+1}\sim N(0,I_M/N)$ independent of the past.  Conversely, the forward map in \cref{lem:one-column-steering} recovers $\Delta B_{k+1}$ from $(b_k,Z_{k+1})$.  Hence adjoining $Z_{k+1}$ or adjoining $\Delta B_{k+1}$ generates the same completed sigma-field.  This proves \eqref{eq:coupling-filtration-equality} at time $k+1$.

The same conditional argument shows that $Z_{k+1}$ is independent of $\cH_0\vee\sigma(Z_1,\ldots,Z_k)$.  Induction therefore also gives the independence of $Z_1,\ldots,Z_N$.
\end{proof}

The restriction to the Brownian mesh is important in \cref{lem:coupling-filtration}: a control depending on Brownian-bridge information inside an interval need not be reconstructible from the Gaussian inputs observed at the mesh points.  \cref{alg:online} uses only mesh values.

\subsection{The one-step coupling error}

Set
\[
    \mathcal K_k
    :=
    \cH_0\vee\cG_k^B.
\]
On $\{b_k\neq0\}$ define
\[
    v_k:=\frac{b_k}{\norm{b_k}_2},
    \qquad
    R_k:=\sqrt N\,\norm{b_k}_2\in[0,\cgauss],
\]
and on $\{b_k=0\}$ take any fixed unit vector for $v_k$ and set $R_k=0$. Also put
\[
    \eta_{k+1}
    :=
    \sqrt N\,\langle\Delta B_{k+1},v_k\rangle.
\]
Conditionally on $\mathcal K_k$, $\eta_{k+1}$ is standard normal.

\begin{lemma}
\label{lem:coupling-one-step}
For the construction in \cref{prop:conditional-coupling,lem:coupling-filtration}, define
\begin{equation}
    \xi_{k+1}
    :=
    -2
    \1_{\{-\Psi(R_k)<\eta_{k+1}\le0\}}
    \eta_{k+1}
    -
    R_k.
\label{eq:coupling-xi-definition}
\end{equation}
Then
\begin{equation}
    (\varepsilon_{k+1}Z_{k+1}-b_k)-\Delta B_{k+1}
    =
    \frac1{\sqrt N}\xi_{k+1}v_k,
\label{eq:coupling-rank-one-residual}
\end{equation}
and
\begin{equation}
    \E[\xi_{k+1}\mid\mathcal K_k]=0.
\label{eq:coupling-xi-centered}
\end{equation}
Moreover, for every $p>2$ there is $C_p<\infty$ such that
\begin{equation}
    \left(
        \E[
            |\xi_{k+1}|^{2p}
            \mid\mathcal K_k
        ]
    \right)^{1/p}
    \le
    C_p\sqrt N\,\norm{b_k}_2.
\label{eq:coupling-xi-moment}
\end{equation}
The constant $C_p$ depends only on $p$.
\end{lemma}

\begin{proof}
The inverse one-column construction leaves the component in the direction $v_k$ unchanged and flips the orthogonal component by the sign $\varepsilon_{k+1}$.  Consequently,
\[
    \varepsilon_{k+1}Z_{k+1}-\Delta B_{k+1}
    =
    (\varepsilon_{k+1}-1)
    \langle\Delta B_{k+1},v_k\rangle v_k.
\]
By \eqref{eq:steering-sign},
\[
    \varepsilon_{k+1}-1
    =
    -2
    \1_{\{-\Psi(R_k)<\eta_{k+1}\le0\}},
\]
while
\[
    b_k
    =
    \frac{R_k}{\sqrt N}v_k.
\]
Substitution gives \eqref{eq:coupling-rank-one-residual} with $\xi_{k+1}$ as in \eqref{eq:coupling-xi-definition}.  The identity defining $\Psi$ in \cref{lem:one-column-steering} gives
\[
    \E\left[
        -2
        \1_{\{-\Psi(R_k)<\eta_{k+1}\le0\}}
        \eta_{k+1}
        \,\middle|\,
        \mathcal K_k
    \right]
    =
    R_k,
\]
which proves \eqref{eq:coupling-xi-centered}.

It remains to bound the conditional moments.  There are constants $r_0\in(0,\cgauss)$ and $C<\infty$ such that
\[
    \Psi(r)\le C\sqrt r,
    \qquad
    0\le r\le r_0.
\]
For $0\le R_k\le r_0$,
\[
    \1_{\{-\Psi(R_k)<\eta_{k+1}\le0\}}
    |\eta_{k+1}|
    \le
    C\sqrt{R_k}.
\]
For $R_k>r_0$, the bound
\[
    |\eta_{k+1}|
    \le
    \frac{\sqrt{R_k}}{\sqrt{r_0}}\,|\eta_{k+1}|
\]
is enough.  Thus, uniformly in $R_k\in[0,\cgauss]$,
\[
    \1_{\{-\Psi(R_k)<\eta_{k+1}\le0\}}
    |\eta_{k+1}|
    \le
    C\sqrt{R_k}\,(1+|\eta_{k+1}|).
\]
Since $R_k$ is $\mathcal K_k$-measurable and $\eta_{k+1}\sim N(0,1)$ conditionally on $\mathcal K_k$,
\[
    \left(
        \E[
            |\xi_{k+1}|^{2p}
            \mid\mathcal K_k
        ]
    \right)^{1/(2p)}
    \le
    C_p\sqrt{R_k}.
\]
Squaring and using $R_k=\sqrt N\,\norm{b_k}_2$ gives \eqref{eq:coupling-xi-moment}.
\end{proof}

\subsection{The conditional \texorpdfstring{$\ell^\infty$}{l-infinity} estimate}

We finish the proof of the coupling bound in \cref{prop:conditional-coupling}.  For a random variable $U$ and $q\ge1$, write
\[
    \norm{U}_{q\mid\cH_0}
    :=
    \left(
        \E[|U|^q\mid\cH_0]
    \right)^{1/q}.
\]
Since $b_k=\norm{b_k}_2v_k$ on $\{b_k\neq0\}$,
\[
    \norm{b_k}_2|v_{k,i}|^2
    =
    |b_{k,i}|\,|v_{k,i}|
    \le
    |b_{k,i}|.
\]
Using the tower property and \eqref{eq:coupling-xi-moment}, we obtain
\begin{align}
    \norm{\xi_{k+1}v_{k,i}}_{2p\mid\cH_0}^{\,2}
    &=
    \left(
        \E\left[
            |v_{k,i}|^{2p}
            \E[
                |\xi_{k+1}|^{2p}
                \mid\mathcal K_k
            ]
            \,\middle|\,
            \cH_0
        \right]
    \right)^{1/p}
    \notag\\
    &\le
    C_p\sqrt N
    \left(
        \E[
            \norm{b_k}_2^p
            |v_{k,i}|^{2p}
            \mid\cH_0
        ]
    \right)^{1/p}
    \notag\\
    &\le
    C_p\sqrt N\,
    \norm{b_{k,i}}_{p\mid\cH_0}.
\label{eq:coupling-coordinate-moment}
\end{align}

Fix $0\le k_0<k_1\le N$ and define
\[
    R_i
    :=
    \frac1{\sqrt N}
    \sum_{k=k_0}^{k_1-1}
    \xi_{k+1}v_{k,i}.
\]
By \eqref{eq:coupling-xi-centered}, the summands are martingale differences with respect to $(\mathcal K_k)$, and the conditional Marcinkiewicz--Zygmund inequality gives
\begin{align}
    \norm{R_i}_{2p\mid\cH_0}^{\,2}
    &\le
    \frac{C_p}{N}
    \sum_{k=k_0}^{k_1-1}
    \norm{\xi_{k+1}v_{k,i}}_{2p\mid\cH_0}^{\,2}
    \notag\\
    &\le
    \frac{C_p}{\sqrt N}
    \sum_{k=k_0}^{k_1-1}
    \norm{b_{k,i}}_{p\mid\cH_0}.
\label{eq:coupling-coordinate-MZ}
\end{align}
The conditional form follows from the usual martingale proof, with all expectations replaced by conditional expectations given $\cH_0$.

Raising \eqref{eq:coupling-coordinate-MZ} to the $p$th power and using
\[
    \left(
        \sum_{k=k_0}^{k_1-1}x_k
    \right)^p
    \le
    N^{p-1}
    \sum_{k=k_0}^{k_1-1}x_k^p
\]
gives
\[
    \E[
        |R_i|^{2p}
        \mid\cH_0
    ]
    \le
    C_pN^{p/2-1}
    \sum_{k=k_0}^{k_1-1}
    \E[
        |b_{k,i}|^p
        \mid\cH_0
    ].
\]
Summing over the coordinates,
\begin{equation}
    \sum_{i=1}^M
    \E[
        |R_i|^{2p}
        \mid\cH_0
    ]
    \le
    C_pN^{p/2}
    \max_{k_0\le k<k_1}
    \sum_{i=1}^M
    \E[
        |b_{k,i}|^p
        \mid\cH_0
    ].
\label{eq:coupling-summed-moment}
\end{equation}
Finally,
\[
    \E[
        \norm{R}_\infty
        \mid\cH_0
    ]
    \le
    \E[
        \norm{R}_{2p}
        \mid\cH_0
    ]
    \le
    \left(
        \sum_{i=1}^M
        \E[
            |R_i|^{2p}
            \mid\cH_0
        ]
    \right)^{1/(2p)}.
\]
Combining this with \eqref{eq:coupling-rank-one-residual} and \eqref{eq:coupling-summed-moment} yields
\[
\begin{split}
&\E\left[
\left\|
\sum_{k=k_0}^{k_1-1}
    (\varepsilon_{k+1}Z_{k+1}-b_k)
-
(B_{t_{k_1}}-B_{t_{k_0}})
\right\|_\infty
\,\middle|\,
\cH_0
\right]\\
&\hspace{1.5cm}
\le
C_pN^{1/4}
\max_{k_0\le k<k_1}
\left(
    \sum_{i=1}^M
    \E[
        |b_{k,i}|^p
        \mid\cH_0
    ]
\right)^{1/(2p)},
\end{split}
\]
which is \eqref{eq:conditional-coupling} and completes the proof of \cref{prop:conditional-coupling}.
\section{F\"ollmer drifts and Euler discretization}
\label{app:follmer}

This appendix constructs the final-block feedback and proves the confinement estimate of \cref{lem:euler-follmer}; it also supplies the F\"ollmer drift used in the large-negative-margin upper bound of \cref{prop:large-negative-upper}.  We first record a general form of the F\"ollmer drift for a compactly supported terminal law.  The statement is deliberately more general than the specialization used by \cref{alg:online}: \cref{prop:large-negative-upper} uses the same result with a different terminal density.  We then construct the shrinking-support terminal law of \cref{lem:compact-follmer}, establish the derivative bounds implied by log-concavity, and prove the Euler estimate.

\subsection{F\"ollmer drift for a compactly supported terminal law}

Let $\gamma_h=N(0,h)$ and let $P_s$ denote the one-dimensional heat semigroup.  The following is the form of the F\"ollmer construction used throughout the paper.  It is the compact-support case of the score-drift identity in~\cite[Lemma~4.1]{FiedlerJacksonLackerNilesWeed26}; we include the details needed when the density vanishes at the boundary of its support.

\begin{lemma}[F\"ollmer drift for a compactly supported density]
\label{lem:localized-follmer}
Let $h>0$, let $p\ge2$, and let $f:\R\to[0,\infty)$ be absolutely continuous and compactly supported, with
\[
    \int f\,d\gamma_h=1.
\]
Define
\[
    s_f(x)
    :=
    \begin{cases}
        f'(x)/f(x),&f(x)>0,\\
        0,&f(x)=0,
    \end{cases}
\]
and suppose $s_f\in L^p(f\,d\gamma_h)$.  For $0\le t<h$, set
\[
    q_t(x):=P_{h-t}f(x),
    \qquad
    v(t,x):=\partial_x\log q_t(x).
\]
Then $q_t(x)>0$ for every $t<h$ and $x\in\R$.

There is a Brownian motion $W$ and a continuous process $X$ satisfying
\begin{equation}
    X_t
    =
    W_t+\int_0^t v(s,X_s)\,ds,
    \qquad 0\le t\le h,
\label{eq:follmer-sde}
\end{equation}
such that $X_h$ has law $f\,d\gamma_h$.  The process $X$ is adapted to the completed natural filtration of $W$, and for $0\le t<h$,
\begin{equation}
    v(t,X_t)
    =
    \E[
        s_f(X_h)
        \mid
        \cF_t^W
    ].
\label{eq:follmer-score-identity}
\end{equation}
Consequently,
\begin{equation}
    \sup_{0\le t<h}
    \norm{v(t,X_t)}_{L^p}
    \le
    \norm{s_f}_{L^p(f\,d\gamma_h)}.
\label{eq:follmer-score-bound}
\end{equation}
\end{lemma}

\begin{proof}
Work first on canonical Wiener space under a probability measure $P$, with coordinate process $B$.  Since the Gaussian kernel is strictly positive, $q_t(x)>0$ for every $t<h$ and $x\in\R$.  Put
\[
    L_t:=q_t(B_t).
\]
By the Markov property,
\[
    L_t
    =
    \E_P[f(B_h)\mid\cF_t^B],
    \qquad
    L_0
    =
    P_hf(0)
    =
    \int f\,d\gamma_h
    =
    1.
\]
Define $Q$ on $\cF_h^B$ by
\[
    \frac{dQ}{dP}=f(B_h).
\]
Fix $T<h$.  On $[0,T]$, localize at stopping times on which $L_t$ stays in a compact subset of $(0,\infty)$ and $|v(t,B_t)|$ remains bounded. It\^o's formula and the backward heat equation give
\[
    dL_t
    =
    L_tv(t,B_t)\,dB_t.
\]
After removing the localization, $L$ is the density process of $Q$ on $\cF_T^B$.  Girsanov's theorem therefore shows that
\[
    W_t
    :=
    B_t-\int_0^t v(s,B_s)\,ds
\]
is $Q$-Brownian on $[0,T]$.

Because $f$ is absolutely continuous and compactly supported, integration by parts gives
\[
    \partial_xP_rf=P_rf'
    \qquad (r>0).
\]
Bayes' formula now yields
\begin{align*}
    \E_Q[s_f(B_h)\mid\cF_t^B]
    &=
    \frac{
        \E_P[
            f(B_h)s_f(B_h)
            \mid
            \cF_t^B
        ]
    }{
        \E_P[
            f(B_h)
            \mid
            \cF_t^B
        ]
    }\\
    &=
    \frac{
        P_{h-t}f'(B_t)
    }{
        P_{h-t}f(B_t)
    }
    =
    v(t,B_t).
\end{align*}
Conditional Jensen gives
\[
    \norm{v(t,B_t)}_{L^p(Q)}
    \le
    \norm{s_f(B_h)}_{L^p(Q)}
    =
    \norm{s_f}_{L^p(f\,d\gamma_h)}.
\]
In particular,
\[
    \E_Q\int_0^h v(t,B_t)^2\,dt<\infty.
\]
Thus $W_t=B_t-\int_0^t v(s,B_s)\,ds$ extends continuously to $t=h$, and the Brownian motions obtained on $[0,T]$ are consistent as $T\uparrow h$.

It remains to verify that the drift in \eqref{eq:follmer-sde} is adapted to the Brownian filtration generated by $W$.  Let the support of $f$ be contained in a bounded interval $[a,b]$.  For $t<h$, differentiating the Gaussian kernel gives
\begin{equation}
    v(t,x)
    =
    \frac{
        \E_{\nu_{t,x}}[Y]-x
    }{
        h-t
    },
\label{eq:follmer-posterior-mean}
\end{equation}
where $\nu_{t,x}$ is the probability measure on $[a,b]$ with density proportional to
\[
    f(y)\exp\!\left(
        -\frac{(y-x)^2}{2(h-t)}
    \right).
\]
Hence, on every interval $[0,T]$ with $T<h$, the function $x\mapsto v(t,x)$ has linear growth uniformly in $t\le T$ and is globally Lipschitz uniformly in $t\le T$.  Indeed,
\[
    \partial_xv(t,x)
    =
    -\frac1{h-t}
    +
    \frac{
        \operatorname{Var}_{\nu_{t,x}}(Y)
    }{
        (h-t)^2
    },
\]
and the variance is at most $(b-a)^2/4$. Therefore the SDE
\[
    dX_t=v(t,X_t)\,dt+dW_t
\]
has a unique strong solution on every $[0,T]$, $T<h$.  Under $Q$, the canonical process $B$ is such a solution, so pathwise uniqueness implies that $B_{[0,T]}$ is measurable with respect to $W_{[0,T]}$.  Conversely, $W$ is defined from $B$, and hence the completed natural filtrations of $B$ and $W$ agree at every time $t<h$.  The score identity above is therefore exactly \eqref{eq:follmer-score-identity} with conditioning on $\cF_t^W$.

Finally, by the definition of $Q$,
\[
    Q(B_h\in dx)
    =
    f(x)\,\gamma_h(dx).
\]
Taking $X=B$ under $Q$ proves all assertions.
\end{proof}

\subsection{A shrinking-support terminal law}

We construct the compactly supported terminal law and its F\"ollmer drift used by the online algorithm.

\begin{lemma}
\label{lem:compact-follmer}
Fix an even integer $p>4$ and $c_F>0$.  For every sufficiently small $h>0$ there are an even, log-concave, compactly supported piecewise-polynomial density $f_h$ with respect to $\gamma_h$ and a F\"ollmer drift
\[
    v_h(t,x):=\partial_x\log P_{h-t}f_h(x),
    \qquad 0\le t<h,
\]
such that the solution of $dZ_t=v_h(t,Z_t)\,dt+dW_t$, $Z_0=0$, satisfies
\begin{equation}
    \sup_{0\le t<h}\norm{v_h(t,Z_t)}_{L^p}\le c_F,
    \qquad
    Z_h\in[-r_h,r_h]\quad\text{a.s.},
\label{eq:follmer-output}
\end{equation}
where $r_h=r_h(p,c_F)\downarrow0$ as $h\downarrow0$.
\end{lemma}

The construction is the shrinking-tent construction from the proof of \cite[Lemma~4.2]{FiedlerJacksonLackerNilesWeed26}, with the scale written for a time interval of length $h$, and the drift is the one supplied by \cref{lem:localized-follmer}.

Fix an even integer $p>4$ and a target bound $c_F>0$, and put
\[
    K:=\frac{c_F}{p}.
\]
For sufficiently small $h>0$, define
\[
    r_h
    :=
    2\sqrt{
        ph\log\!\left(
            4K^{-2}/h
        \right)
    }
\]
and
\[
    H_h(x)
    :=
    \begin{cases}
        1,&|x|\le r_h/2,\\
        2(r_h-|x|)/r_h,&r_h/2<|x|<r_h,\\
        0,&|x|\ge r_h.
    \end{cases}
\]
The Gaussian tail estimate
\[
    \gamma_h([x,\infty))
    \le
    e^{-x^2/(2h)}
    \qquad (x\ge0)
\]
gives
\begin{equation}
    \gamma_h([r_h/2,\infty))
    \le
    e^{-r_h^2/(8h)}
    =
    \left(\frac K2\right)^p h^{p/2}.
\label{eq:tent-tail}
\end{equation}
For all sufficiently small $h$,
\[
    \gamma_h([0,r_h/2])\ge\frac14,
    \qquad
    r_h\ge2\sqrt h.
\]
Since $H_h'=0$ on $[0,r_h/2]$ and $|H_h'|=2/r_h$ on $(r_h/2,r_h)$,
\[
    \norm{H_h'}_{L^p(\gamma_h)}^p
    =
    2\left(\frac2{r_h}\right)^p
    \gamma_h([r_h/2,r_h]),
\]
whereas
\[
    \norm{H_h}_{L^p(\gamma_h)}^p
    \ge
    2\gamma_h([0,r_h/2]).
\]
Using \eqref{eq:tent-tail}, together with $r_h\ge2\sqrt h$ and $\gamma_h([0,r_h/2])\ge1/4$, yields
\[
    \gamma_h([r_h/2,r_h])
    \le
    \left(
        \frac{Kr_h}{2}
    \right)^p
    \gamma_h([0,r_h/2]).
\]
Consequently,
\begin{equation}
    \norm{H_h'}_{L^p(\gamma_h)}
    \le
    K\norm{H_h}_{L^p(\gamma_h)}.
\label{eq:tent-gradient-bound}
\end{equation}

Set
\[
    f_h
    :=
    \frac{
        H_h^p
    }{
        \int H_h^p\,d\gamma_h
    }.
\]
The function $H_h$ is even and log-concave, and hence so is $f_h$.  On the interior of its support,
\[
    s_{f_h}
    =
    p\,\frac{H_h'}{H_h}.
\]
Therefore \eqref{eq:tent-gradient-bound} gives
\[
    \int
    |s_{f_h}|^p
    f_h\,d\gamma_h
    =
    p^p
    \frac{
        \int|H_h'|^p\,d\gamma_h
    }{
        \int H_h^p\,d\gamma_h
    }
    \le
    c_F^p.
\]
Applying \cref{lem:localized-follmer} proves the $L^p$ bound and the terminal-support assertion in \cref{lem:compact-follmer}.  Finally,
\[
    r_h
    =
    O\!\left(
        \sqrt{h\log(1/h)}
    \right)
    \longrightarrow0,
\]
which completes the proof of that lemma.

The feedback evaluated by the online algorithm is $v_N=v_{h_N}$, which by the definition of $P_{h_N-t}$ has the explicit form
\begin{equation}
 v_N(t,x)=
 \frac{\displaystyle
 \int_{-r_{h_N}}^{r_{h_N}}
 \frac{y-x}{s}f_{h_N}(y)e^{-(y-x)^2/(2s)}\,dy}
 {\displaystyle
 \int_{-r_{h_N}}^{r_{h_N}}
 f_{h_N}(y)e^{-(y-x)^2/(2s)}\,dy},
 \qquad s=h_N-t>0.
\label{eq:explicit-follmer-feedback}
\end{equation}
Because $f_{h_N}$ is piecewise polynomial, the numerator and denominator are finite sums of truncated Gaussian moments; no continuous-time SDE solver is used by the online algorithm.

\subsection{Derivative bounds for the feedback}

The stability of the Euler scheme follows from log-concavity.  We state the derivative estimate in a form that makes the nonexpansiveness of one Euler step immediate.

\begin{lemma}
\label{lem:follmer-derivative-bounds}
Let $f:\R\to[0,\infty)$ be log-concave and not identically zero, and set
\[
    v(t,x)
    :=
    \partial_x\log P_{h-t}f(x),
    \qquad 0\le t<h.
\]
Then
\[
    -\frac1{h-t}
    \le
    \partial_xv(t,x)
    \le0
\]
for every $t<h$ and $x\in\R$.  If $f$ is even, then
\[
    v(t,0)=0,
    \qquad
    |v(t,x)|
    \le
    \frac{|x|}{h-t}.
\]
\end{lemma}

\begin{proof}
Gaussian convolution preserves log-concavity \cite{Prekopa71}.  Hence
\[
    \partial_{xx}\log P_{h-t}f(x)\le0.
\]
For the opposite bound, put $\tau=h-t$ and let $\nu_{x,\tau}$ be the probability measure with density proportional to
\[
    f(y)e^{-(y-x)^2/(2\tau)}.
\]
Differentiating the Gaussian kernel gives
\[
    \partial_{xx}\log P_\tau f(x)
    =
    -\frac1\tau
    +
    \frac{
        \operatorname{Var}_{\nu_{x,\tau}}(Y)
    }{
        \tau^2
    }
    \ge
    -\frac1\tau.
\]
Since $\partial_xv=\partial_{xx}\log P_{h-t}f$, the first assertion follows.

If $f$ is even, so is $P_{h-t}f$, and therefore $v(t,0)=0$. Integrating the derivative bound between $0$ and $x$ yields
\begin{equation}
    |v(t,x)|
    \le
    \frac{|x|}{h-t}.
\label{eq:follmer-linear-bound}
\end{equation}
\end{proof}

In particular, for the feedback $v_h$ in \cref{lem:compact-follmer},
\[
    -\frac1{h-t}
    \le
    \partial_xv_h(t,x)
    \le0,
\]
which controls one Euler step.  If $\Delta=1/N$ and $h-t\ge\Delta$, the map
\[
    x\longmapsto x+\Delta v_h(t,x)
\]
is nondecreasing and $1$-Lipschitz.

We also record the elementary deterministic estimate used to accumulate the local discretization errors.

\begin{lemma}
\label{lem:dissipative-recursion}
Let $e_0=P_0=0$ and suppose
\[
    e_{n+1}
    =
    \theta_ne_n+(P_{n+1}-P_n),
    \qquad
    0\le\theta_n\le1.
\]
Then
\[
    \max_{0\le n\le L}|e_n|
    \le
    2\max_{0\le n\le L}|P_n|.
\]
\end{lemma}

\begin{proof}
For fixed $n$, write
\[
    \theta_{k:n}
    :=
    \prod_{m=k}^{n-1}\theta_m,
    \qquad
    \theta_{n:n}:=1.
\]
Summation by parts gives
\[
    e_n
    =
    P_n
    -
    \sum_{k=1}^{n-1}
    \bigl(
        \theta_{k+1:n}-\theta_{k:n}
    \bigr)P_k.
\]
The coefficients in the sum are nonnegative and have total mass $1-\theta_{1:n}\le1$.  The claimed bound follows.
\end{proof}

\subsection{Euler discretization}

We now prove \cref{lem:euler-follmer}.  Let $C_p^{\circ}\ge1$ be a constant, depending only on $p$, large enough to dominate the constants in the martingale and truncation estimates below.  Run the construction of \cref{lem:compact-follmer} with the reduced budget
\[
    \widetilde c_F:=\frac{c_F}{2C_p^{\circ}},
\]
take $v_N=v_{h_N}$ for the resulting feedback, and set
\[
    s_p(h,c_F):=r_h\!\left(p,\widetilde c_F\right),
\]
so that the terminal support radius of \cref{lem:compact-follmer} equals $r_{h_N}=s_p(h_N,c_F)$; as an explicit elementary function of $h$, $s_p(\cdot,c_F)$ is continuous on $(0,h_0)$.  Put
\[
    \Delta:=\frac1N
\]
and let $Z_i$ denote the continuous F\"ollmer process with the same Brownian driver as the Euler scheme \eqref{eq:euler-follmer}.  For $t<h_N$, define
\[
    M_i(t)
    :=
    v_N(t,Z_i(t)),
\]
and at the endpoint set
\[
    M_i(h_N)
    :=
    s_{f_{h_N}}(Z_i(h_N)).
\]
By \eqref{eq:follmer-score-identity}, $(M_i(t))_{0\le t\le h_N}$ is a martingale.  Since $f_{h_N}$ is even, $M_i(0)=0$, and by the choice of $\widetilde c_F$,
\[
    \norm{M_i(h_N)}_{L^p}
    \le
    \widetilde c_F.
\]
By martingale representation, write
\[
    M_i(t)
    =
    \int_0^t\zeta_i(s)\,dB_i(s).
\]

For $0\le k<L_N$, define the local truncation error
\[
    \delta_{k,i}
    :=
    \int_{k\Delta}^{(k+1)\Delta}
    \bigl(
        M_i(t)-M_i(k\Delta)
    \bigr)\,dt.
\]
Stochastic Fubini gives
\begin{equation}
    \delta_{k,i}
    =
    \int_{k\Delta}^{(k+1)\Delta}
    \bigl(
        (k+1)\Delta-s
    \bigr)
    \zeta_i(s)\,dB_i(s).
\label{eq:euler-local-error}
\end{equation}
Let
\[
    e_{k,i}
    :=
    \widehat Z_{k,i}-Z_i(k\Delta).
\]
Comparing the exact and Euler updates and applying the mean-value theorem,
\[
    e_{k+1,i}
    =
    \theta_{k,i}e_{k,i}
    -
    \delta_{k,i},
\]
where
\[
    \theta_{k,i}
    =
    1+\Delta
    \partial_xv_N(k\Delta,\xi_{k,i})
\]
for a point $\xi_{k,i}$ between $\widehat Z_{k,i}$ and $Z_i(k\Delta)$.  Since
\[
    h_N-k\Delta
    =
    (L_N-k)\Delta
    \ge\Delta,
\]
\cref{lem:follmer-derivative-bounds} gives
\[
    0\le\theta_{k,i}\le1.
\]

Let
\[
    P_{n,i}
    :=
    -\sum_{k=0}^{n-1}\delta_{k,i}.
\]
By \eqref{eq:euler-local-error}, $P_{n,i}$ is the value at $n\Delta$ of the martingale
\[
    \int_0^t
    w_N(s)\zeta_i(s)\,dB_i(s),
\]
where on the interval $[k\Delta,(k+1)\Delta)$,
\[
    w_N(s):=s-(k+1)\Delta,
    \qquad
    |w_N(s)|\le\Delta.
\]
The preceding recursion and \cref{lem:dissipative-recursion} therefore imply
\[
    \max_{j\le L_N}|e_{j,i}|
    \le
    2\max_{j\le L_N}|P_{j,i}|.
\]
Doob's maximal inequality and the Burkholder--Davis--Gundy inequality \cite{RevuzYor99} give
\[
\begin{split}
    \norm{
        \max_{j\le L_N}|P_{j,i}|
    }_{L^p}
    &\le
    C_p
    \norm{
        \left(
            \int_0^{h_N}
            w_N(s)^2\zeta_i(s)^2\,ds
        \right)^{1/2}
    }_{L^p}\\
    &\le
    C_p\Delta
    \norm{
        \left(
            \int_0^{h_N}
            \zeta_i(s)^2\,ds
        \right)^{1/2}
    }_{L^p}.
\end{split}
\]
The reverse Burkholder--Davis--Gundy inequality applied to $M_i(h_N)=\int_0^{h_N}\zeta_i(s)\,dB_i(s)$ gives
\[
    \norm{
        \left(
            \int_0^{h_N}
            \zeta_i(s)^2\,ds
        \right)^{1/2}
    }_{L^p}
    \le
    C_p
    \norm{M_i(h_N)}_{L^p}.
\]
Combining these estimates and using the choice of $C_p^\circ$, we obtain
\begin{equation}
    \norm{
        \max_{j\le L_N}|e_{j,i}|
    }_{L^p}
    \le
    C_p^\circ\Delta
    \norm{M_i(h_N)}_{L^p}
    \le
    C_p^\circ\Delta\widetilde c_F.
\label{eq:euler-max-error}
\end{equation}

We next bound the drift evaluated by the Euler scheme.  By \cref{lem:follmer-derivative-bounds},
\[
\begin{split}
    \norm{
        \widehat a_{k,i}-M_i(k\Delta)
    }_{L^p}
    &=
    \norm{
        v_N(k\Delta,\widehat Z_{k,i})
        -
        v_N(k\Delta,Z_i(k\Delta))
    }_{L^p}\\
    &\le
    \frac{
        \norm{e_{k,i}}_{L^p}
    }{
        h_N-k\Delta
    }\\
    &\le
    C_p^\circ\widetilde c_F.
\end{split}
\]
Since $M_i$ is a martingale, conditional Jensen gives
\[
    \norm{M_i(k\Delta)}_{L^p}
    \le
    \norm{M_i(h_N)}_{L^p}
    \le
    \widetilde c_F.
\]
By the triangle inequality with the drift-error bound above, and since $\widetilde c_F=c_F/(2C_p^\circ)$ with $C_p^\circ\ge1$,
\[
    \sup_{k<L_N}
    \norm{\widehat a_{k,i}}_{L^p}
    \le
    (1+C_p^\circ)\widetilde c_F
    \le
    c_F,
\]
which is \eqref{eq:euler-request-bound}.

Finally, the continuous endpoint satisfies
\[
    Z_i(h_N)
    \in
    \left[
        -s_p(h_N,c_F),
        s_p(h_N,c_F)
    \right]
    \qquad\text{a.s.}
\]
by \cref{lem:compact-follmer}.  Hence
\[
\begin{split}
    \E
    \max_{i\le M}
    \dist\!\left(
        \widehat Z_{L_N,i},
        [-s_p(h_N,c_F),s_p(h_N,c_F)]
    \right)
    &\le
    \E\max_{i\le M}|e_{L_N,i}|\\
    &\le
    \left(
        \sum_{i=1}^M
        \E|e_{L_N,i}|^p
    \right)^{1/p}\\
    &=
    O(M^{1/p}N^{-1}),
\end{split}
\]
where the implicit constant depends only on $p$ and $c_F$ on a fixed compact range of $h_N$.  This proves \eqref{eq:euler-endpoint-bound} and completes the proof of \cref{lem:euler-follmer}.
\section{Finite-precision implementation}
\label{app:finite-precision}

The $O(MN)$ bound in the main text counts arithmetic operations.  This appendix verifies that the algorithm does not rely on exact real arithmetic.  For every fixed $(\alpha,\kappa)$ below the threshold and every fixed choice of parameters satisfying the strict inequalities in \cref{prop:finite-transfer}, we give a dyadic implementation with bit complexity polynomial in $M+N$ while preserving a positive terminal margin.

Two points require care.  First, the steering rule \eqref{eq:steering-sign} is discontinuous at its switching thresholds, so a finite-precision implementation may fail to determine the exact branch when the Gaussian coordinate is very close to a threshold.  Gaussian anti-concentration makes the probability of any such ambiguity negligible. Second, the final-stage feedback \eqref{eq:explicit-follmer-feedback} is a ratio of truncated Gaussian integrals.  We show that every feedback value required by the algorithm can be evaluated to inverse-polynomial accuracy using polynomially many bits.

We use the following input model.  Each Gaussian coordinate is an exact random variable with distribution $N(0,1)$, accessed through an oracle that, on input $L\in\N$, returns a dyadic approximation with absolute error at most $2^{-L}$ in time polynomial in $L$.  Bit complexity counts ordinary bit operations together with the number of oracle bits requested.  This model separates the probabilistic input distribution from its numerical representation.

\begin{theorem}
\label{thm:finite-precision}
Fix $\kappa\in\R$ and $\alpha<\aon(\kappa)$.  Fix parameters for \cref{alg:online} satisfying the strict inequalities in \cref{prop:finite-transfer}, with terminal level $r>\kappa$.  Then the algorithm has a dyadic implementation using a number of bit operations polynomial in $M+N$ and satisfying
\[
    \Pp\left(
        \min_{1\le i\le M}X_{N,i}
        \ge
        \frac{\kappa+r}{2}
    \right)
    \longrightarrow1
\]
whenever $M/N\to\alpha$.
\end{theorem}

The numerical exponents below are chosen only for convenience.  No attempt is made to optimize them.

\subsection{Numerical implementation}

Set
\begin{equation}
    \tau_N:=N^{-4},
    \qquad
    \varepsilon_N:=N^{-6}.
\label{eq:precision-scales}
\end{equation}
Here $\tau_N$ provides an inverse-polynomial guard from the boundary of the one-column drift ball, and $\varepsilon_N$ is the accuracy used for stored states, scalar feedback evaluations, and the comparisons of the scalar Gaussian coordinate with the switching thresholds in \eqref{eq:steering-sign}.

We implement \cref{alg:online} as follows.

\begin{enumerate}[label=(\roman*),leftmargin=2.5em]
\item Every stored coordinate of $X$, $B$, and $\widehat Z$ is rounded after each update to a dyadic number with absolute error at most $\varepsilon_N$.

\item The coefficient functions defining the first-stage simple predictable control and the final-stage feedback $v_N$ are evaluated to absolute error at most $\varepsilon_N$.

\item Before applying the one-column signing rule, the raw control vector $q^{\rm raw}$ is replaced by a dyadic vector $q^\#$ satisfying
\begin{equation}
    \norm{q^\#}_2
    \le
    (\cgauss-\tau_N)\sqrt N.
\label{eq:safe-capacity}
\end{equation}
If
\[
    \norm{q^{\rm raw}}_2
    \le
    (\cgauss-3\tau_N)\sqrt N
\]
and $\norm{q^{\rm raw}}_2>2\tau_N$, the approximation is chosen so that
\begin{equation}
    \norm{q^\#-q^{\rm raw}}_\infty
    \le
    \varepsilon_N,
    \qquad
    \norm{q^\#}_2\ge\tau_N.
\label{eq:request-rounding}
\end{equation}
If the norm of $q^{\rm raw}$ cannot be certified to exceed $2\tau_N$, we set $q^\#=0$.  Outward-rounded norm evaluation, coordinatewise dyadic rounding, and, when necessary, radial contraction produce such a vector in polynomial time.

\item For $q^\#\neq0$, put
\[
    R:=\frac{\norm{q^\#}_2}{\sqrt N},
    \qquad
    v:=\frac{q^\#}{\norm{q^\#}_2}.
\]
Using interval arithmetic, evaluate both switching levels in \eqref{eq:steering-sign} and the scalar projection $\langle g_k,v\rangle$ until the resulting intervals have width at most $\varepsilon_N$.  If the intervals determine the branch of \eqref{eq:steering-sign}, use that branch.  Otherwise choose $+1$ and use a dyadic approximation to $g_k/\sqrt N$ as the stored Gaussian increment.

\item When the branch is determined, evaluate \eqref{eq:brownian-innovation} coordinatewise to absolute error at most $\varepsilon_N$ before updating the stored Brownian state.
\end{enumerate}

The lower cutoff in (iii) avoids dividing by a vector whose norm is too small. The contraction in \eqref{eq:safe-capacity} keeps the logarithm defining $\Psi$ away from its singular endpoint.  Indeed, for every nonzero $q^\#$,
\[
    1-\frac{R}{\cgauss}
    \ge
    \frac{\tau_N}{\cgauss}
    =
    \frac{1}{\cgauss N^4}.
\]
Hence
\[
    \Psi(R)
    =
    O(\sqrt{\log N}),
\]
and the switching levels can be enclosed to width $\varepsilon_N$ with polynomially many bits.  The normalization by $\norm{q^\#}_2$ is also well-conditioned at inverse-polynomial precision because $\norm{q^\#}_2\ge\tau_N$ whenever $q^\#\neq0$.

\subsection{Polynomial bounds on the numerical state}

Let
\[
    \mathcal A_N
    :=
    \left\{
        \max_{1\le k\le N}
        \max_{1\le i\le M}
        |g_{k,i}|
        \le N
    \right\}.
\]
For fixed $\alpha\in(0,\infty)$ and $M/N\to\alpha$, a Gaussian tail bound gives
\begin{equation}
    \Pp(\mathcal A_N^c)
    \le
    C N^2e^{-N^2/2}
    =
    o(1).
\label{eq:gaussian-input-box}
\end{equation}

\begin{lemma}
\label{lem:state-box}
On $\mathcal A_N$, every exact state used to couple the numerical implementation to \cref{alg:online}, every stored dyadic state, and every unprojected scalar control value has absolute value at most $N^4$ for all sufficiently large $N$.
\end{lemma}

\begin{proof}
Each coordinate increment of the normalized margin $X$ has magnitude at most $\sqrt N$ on $\mathcal A_N$, so
\[
    \norm{X}_\infty=O(N^{3/2}).
\]
For $Z=g_k/\sqrt N$, each coordinate has magnitude at most $\sqrt N$ and, since $M=O(N)$,
\[
    \norm{Z}_2=O(N).
\]
The transformation in \eqref{eq:brownian-innovation} is orthogonal on the component perpendicular to the current drift direction and leaves the parallel component unchanged.  Thus
\[
    \norm{\Delta B}_2=\norm{Z}_2=O(N),
\]
and the accumulated Brownian state is bounded by $O(N^2)$ in $\ell^\infty$.

In the final stage, \cref{lem:follmer-derivative-bounds} implies that each Euler map
\[
    x\longmapsto x+\frac1N v_N(k/N,x)
\]
is nondecreasing, $1$-Lipschitz, and fixes zero.  Consequently the exact Euler state is bounded by the accumulated Brownian increments, and the stored state differs only by the accumulated rounding errors.  Both are therefore $O(N^2)$ in absolute value.  Moreover, \eqref{eq:follmer-linear-bound} gives, at a mesh time,
\[
    |v_N(k/N,x)|
    \le
    N|x|
    =
    O(N^3).
\]
The deterministic correction
\[
    q_i^{\rm det}
    =
    \frac{(a-X_i)_+}{h_N}
\]
is $O(N^{3/2})$ because $h_N$ converges to the fixed positive number $h$. The first-stage control is bounded by definition.  All quantities above are therefore at most $N^4$ for sufficiently large $N$.
\end{proof}

The only scalar evaluation whose denominator can become exponentially small is the final-stage feedback in \eqref{eq:explicit-follmer-feedback}.  The next lemma shows that even the crude polynomial state bound above is enough for polynomial bit complexity.

\begin{lemma}
\label{lem:feedback-evaluation}
On the event $\mathcal A_N$, every value of $v_N(t,x)$ required by \cref{alg:online} can be evaluated to absolute error $\varepsilon_N$ using polynomially many bits of precision and polynomially many bit operations.
\end{lemma}

\begin{proof}
For all sufficiently large $N$, the final-stage length $h_N$ lies in a fixed compact neighborhood of the fixed number $h>0$.  The density $f_{h_N}$ defining $v_N$ is bounded below on an interval $J$ contained in the interior of its plateau, with the length of $J$ and the lower bound uniform over that neighborhood of $h$.

Write $s=h_N-t$.  At every mesh point at which the feedback is evaluated, $s\ge1/N$.  On the event $\mathcal A_N$, \cref{lem:state-box} gives a polynomial bound on $|x|$.  Hence there are constants $c,C,d>0$, independent of $N$, such that the denominator in \eqref{eq:explicit-follmer-feedback} satisfies
\[
    \int
    f_{h_N}(y)e^{-(y-x)^2/(2s)}\,dy
    \ge
    c\int_J
    e^{-(y-x)^2/(2s)}\,dy
    \ge
    \exp(-C N^d).
\]

By \eqref{eq:follmer-linear-bound}, $|v_N(t,x)|\le|x|/s\le N^{O(1)}$ at every such mesh point, since $s\ge1/N$ and \cref{lem:state-box} bounds $|x|$ polynomially.  As the numerator of \eqref{eq:explicit-follmer-feedback} equals $v_N(t,x)$ times the denominator, it is at most $N^{O(1)}$ times the denominator.

Because $f_{h_N}$ is piecewise polynomial, the numerator and denominator of \eqref{eq:explicit-follmer-feedback} reduce on each piece to finite linear combinations of truncated Gaussian moments.  The endpoints and coefficients, including the quantities obtained from $h_N$ by elementary functions, can be computed to any prescribed precision in time polynomial in that precision. Since the denominator is at least $\exp(-CN^d)$ and the numerator at most a polynomial multiple of it, evaluating each term with directed rounding to absolute error $\exp(-2CN^d)\varepsilon_N$ determines the quotient to absolute error $\varepsilon_N$, using polynomially many bits.  Standard arbitrary-precision algorithms for rational arithmetic and elementary functions have bit complexity polynomial in the requested precision; see, for example, \cite{FousseHanrotLefevrePelissierZimmermann07}.  This proves the claim.
\end{proof}

\subsection{Threshold comparisons and rounding stability}

The only discontinuous operation in the one-column signing map is the comparison with the two switching levels in \eqref{eq:steering-sign}. Gaussian anti-concentration makes the event that finite precision cannot resolve one of these comparisons negligible.

\begin{lemma}
\label{lem:switching-ambiguity}
Condition on the past of the numerical algorithm and suppose $q^\#\neq0$. The probability that the interval computation in step~(iv) does not determine the branch of \eqref{eq:steering-sign} is at most $C\varepsilon_N$, uniformly in the past.  Consequently, the probability that any such ambiguity occurs during the entire execution is
\[
    O(N\varepsilon_N)
    =
    O(N^{-5}).
\]
\end{lemma}

\begin{proof}
Conditionally on the past,
\[
    \eta
    :=
    \left\langle
        g_k,
        \frac{q^\#}{\norm{q^\#}_2}
    \right\rangle
\]
is standard normal.  The exact sign can change only when $\eta$ crosses one of the two past-measurable levels
\[
    0,
    \qquad
    -\Psi\!\left(
        \frac{\norm{q^\#}_2}{\sqrt N}
    \right).
\]
If the interval computation fails to determine the branch, $\eta$ lies within $C\varepsilon_N$ of one of these levels.  The standard normal density is uniformly bounded, so this event has conditional probability $O(\varepsilon_N)$.  A union bound over the $N$ arrivals proves the second claim.
\end{proof}

On the event that every threshold comparison is resolved, we couple the numerical implementation to the exact Gaussian steering map using the actual dyadic control vectors $q^\#$.  Let $B_k^\#$ denote the cumulative exact Gaussian increments produced by that map, and let $\widetilde B_k$ denote the stored dyadic Brownian state.

\begin{lemma}
\label{lem:first-stage-rounding}
On the event that all threshold comparisons are resolved,
\begin{equation}
    \max_{k\le k_0}
    \norm{
        \widetilde B_k-B_k^\#
    }_\infty
    =
    O(N\varepsilon_N).
\label{eq:first-stage-state-error}
\end{equation}
Until the guarded projection in \eqref{eq:safe-capacity} is activated, the dyadic first-stage control differs coordinatewise from the exact simple predictable control evaluated on $B^\#$ by
\begin{equation}
    O(N\varepsilon_N+\tau_N).
\label{eq:first-stage-request-error}
\end{equation}
The resulting error in the accumulated predictable drift is of the same order.
\end{lemma}

\begin{proof}
For a fixed dyadic control vector and a resolved threshold comparison, the exact and numerical updates use the same sign.  The numerical Gaussian increment is evaluated and stored with coordinatewise error $O(\varepsilon_N)$, so summing over at most $N$ steps gives \eqref{eq:first-stage-state-error}.

The first-stage control is defined by finitely many globally Lipschitz coefficient functions of finitely many stored Brownian values.  Replacing the exact mesh values by the stored values therefore changes each scalar control value by $O(N\varepsilon_N)$.  Numerical evaluation contributes an additional $O(\varepsilon_N)$.  If the lower norm cutoff is activated, the discarded vector has Euclidean norm $O(\tau_N)$ and hence changes every coordinate by $O(\tau_N)$.  Before the guarded projection is active, these are the only errors, proving \eqref{eq:first-stage-request-error}.

A control vector enters the predictable drift with factor $1/N$ at one step.  Summing the coordinatewise errors over at most $N$ first-stage steps therefore preserves the order $O(N\varepsilon_N+\tau_N)$.
\end{proof}

For the final stage, let $Z_k^\#$ be the exact Euler scheme driven by the exact Gaussian increments associated with the dyadic control vectors, and let $\widetilde Z_k$ be the stored dyadic state.

\begin{lemma}
\label{lem:repair-rounding}
On the event that all threshold comparisons are resolved and until the guarded projection in \eqref{eq:safe-capacity} is activated,
\begin{equation}
    \max_{k\le L_N}
    \norm{
        \widetilde Z_k-Z_k^\#
    }_\infty
    =
    O(N\varepsilon_N).
\label{eq:repair-state-error}
\end{equation}
At final-stage step $k$, the dyadic control differs coordinatewise from
\[
    q^{\rm det}+v_N(k/N,Z_k^\#)
\]
by
\begin{equation}
    O(N^2\varepsilon_N+\tau_N).
\label{eq:repair-request-error}
\end{equation}
\end{lemma}

\begin{proof}
By \cref{lem:follmer-derivative-bounds}, the exact Euler map
\[
    x\longmapsto
    x+\frac1N v_N(k/N,x)
\]
is $1$-Lipschitz at every final-stage mesh point.  The error in the evaluated feedback contributes $O(\varepsilon_N/N)$ to one state update, the evaluated Gaussian increment contributes $O(\varepsilon_N)$, and storing the updated state contributes another $O(\varepsilon_N)$.  Thus the coordinatewise state error satisfies
\[
    e_{k+1}\le e_k+C\varepsilon_N,
\]
which gives \eqref{eq:repair-state-error}.

At a mesh point,
\[
    |\partial_xv_N|
    \le N
\]
by \cref{lem:follmer-derivative-bounds}.  The state error $O(N\varepsilon_N)$ therefore induces an $O(N^2\varepsilon_N)$ error in the feedback value.  The deterministic correction is computed from the stored margin $X$; since that state has accumulated only $O(N\varepsilon_N)$ coordinatewise rounding error and $h_N$ is bounded away from zero, its contribution is smaller.  Numerical evaluation contributes $O(\varepsilon_N)$, and the lower norm cutoff contributes $O(\tau_N)$.  This proves \eqref{eq:repair-request-error}.
\end{proof}

\subsection{Proof of the finite-precision theorem}

\begin{proof}[Proof of \cref{thm:finite-precision}]
Let $\mathcal E_N$ be the event that $\mathcal A_N$ holds and every threshold comparison is resolved.  By \cref{eq:gaussian-input-box,lem:switching-ambiguity},
\[
    \Pp(\mathcal E_N)\longrightarrow1.
\]
On $\mathcal E_N$, each one-column update of the numerical algorithm is coupled to the exact Gaussian signing map driven by its actual dyadic control vector.  Thus the only differences from the exact analysis of \cref{prop:finite-transfer} are the deterministic numerical errors quantified above.

The first- and final-stage stability estimates give the uniform error scale
\begin{equation}
    N^2\varepsilon_N+\tau_N
    =
    O(N^{-4})
    =
    o(1).
\label{eq:total-rounding-error}
\end{equation}
The hypotheses of \cref{prop:finite-transfer} contain only strict inequalities.  In particular, there are fixed positive gaps in
\[
    \norm{\beta}_{\infty,2}<\ca,
    \qquad
    \frac1h\norm{(a-Y_T^\beta)_+}_{L^p}+c_F<\ca,
    \qquad
    s_p(h,c_F)<a-r.
\]
Therefore the perturbation \eqref{eq:total-rounding-error} is absorbed by these gaps for all sufficiently large $N$.

It follows first that, with probability $1-o(1)$, every exact raw control vector in the proof of \cref{prop:finite-transfer} lies a fixed distance inside the one-column drift ball.  Since $\tau_N\to0$, the guarded projection \eqref{eq:safe-capacity} is then never activated.  The lower norm cutoff changes the accumulated predictable drift by at most $O(\tau_N)$ and is also absorbed by \eqref{eq:total-rounding-error}.  The Brownian comparison, empirical shortfall estimate, final-stage confinement, and conditional Gaussian coupling from the proof of \cref{prop:finite-transfer} therefore remain valid after an additional $o(1)$ perturbation.

The exact arithmetic analysis gives terminal margin at least $r$ with probability tending to one.  The stored normalized margin differs from the exact signed sum generated by the numerical signs by at most
\[
    O(N\varepsilon_N)=o(1)
\]
coordinatewise.  Since $r-\kappa>0$ is fixed, for all sufficiently large $N$ the total numerical error is smaller than $(r-\kappa)/2$.  Hence
\[
    \min_{1\le i\le M}X_{N,i}
    \ge
    \frac{\kappa+r}{2}
\]
with probability tending to one.

It remains to bound the running time.  On $\mathcal A_N$, all stored states and scalar control values have polynomial magnitude by \cref{lem:state-box}.  The first-stage coefficient functions have fixed rational piecewise-polynomial descriptions, so they can be evaluated to $\varepsilon_N$ accuracy with polynomially many bit operations.  Vector norms, inner products, radial contractions, and the Gaussian steering map require only polynomially many bits because every nonzero dyadic control vector has norm at least $\tau_N$ and remains at least $\tau_N\sqrt N$ in radial distance from the singular endpoint of the drift ball.  By \cref{lem:feedback-evaluation}, each final-stage feedback evaluation uses polynomially many bits and polynomially many bit operations.

There are $O(MN)$ scalar state updates and feedback evaluations.  Since $M/N\to\alpha$, the product of this count with any fixed polynomial bound for one scalar operation remains polynomial in $M+N$.  The total bit complexity is therefore polynomial in $M+N$ on $\mathcal A_N$.

For a bound on every realization of the input oracle, run the implementation under a fixed polynomial budget $P$, depending only on the fixed algorithmic parameters, on the requested oracle precision, the stored-state magnitude, and the number of bit operations, taking $P$ larger than the bounds above; if the budget would be exceeded---in particular if some $|g_{k,i}|$ cannot be certified below $N$---output $+1$ for all remaining signs.  Set $\mathcal A_N':=\{\max_{k,i}|g_{k,i}|\le N/2\}$.  Since $\mathcal A_N'\subseteq\mathcal A_N$ and each coordinate is then certified below $N$ with $O(\log N)$ bits, the cap is never reached on $\mathcal A_N'$ for all large $N$.  Intersecting with the event that all switching comparisons of \eqref{eq:steering-sign} are resolved, of probability $1-o(1)$ by \cref{lem:switching-ambiguity}, the capped implementation agrees with the execution analyzed above.  Since $\Pp((\mathcal A_N')^c)\le CN^2e^{-N^2/8}=o(1)$, the terminal-margin guarantee is unchanged, while the budget makes the bit complexity polynomial in $M+N$ on every input.
\end{proof}

    \section{Computer-assisted verification at zero margin}
\label{app:computer-assisted}

This appendix proves the two computer-assisted statements used in \cref{sec:zero-margin}.  The upper-bound verifier checks the analytic inequalities in \cref{lem:certified-profile}; the lower-bound verifier checks a global upper enclosure for the dynamic-programming recursion in \cref{thm:weighted-recursion}.  These are logically independent computations.

In both cases, the search procedure used to discover the supplied data is outside the trusted argument.  The mathematical proof consists of the soundness statements below together with acceptance of the fixed certificate data by the corresponding verifier.

\subsection{Verification of the zero-margin upper bound}
\label{app:upper-certificate}

The certificate for \cref{lem:certified-profile} represents the function $g$ by a rational piecewise-polynomial spline on $[-4,4]$, joined with matching values and first derivatives to the analytic tails
\[
    g(z)
    =
    a(-z)+d(-z)^{-1/2},
    \qquad z\le-4,
\]
and
\[
    g(z)
    =
    z(w+w^2),
    \qquad
    w=Ae^{-(z^2-16)/2},
    \qquad z\ge4.
\]
The constants $a,d,A$ and all spline coefficients are rational.  The verifier checks exact matching of $g$ and $g'$ at every interface.

\paragraph{Profile inequalities.}
On every interval on which $g$ is smooth, define
\[
\begin{aligned}
    Q_0&:=g,
    &
    Q_1&:=g+z,
    &
    Q_2&:=\cL g+\Lambda g,
\\
    Q_3&:=-g-\cL g,
    &
    Q_4&:=\mathfrak r,
    &
    Q_5&:=\cL\mathfrak r+\lambda\mathfrak r,
\end{aligned}
\qquad
    \mathfrak r:=Cg-(g')^2.
\]
Thus the pointwise assertions in \eqref{eq:profile-inequalities} are exactly
\[
    Q_0>0,
    \qquad
    Q_j\ge0
    \quad (1\le j\le5).
\]

For each spline piece, an affine change of variables maps its domain to $[0,1]$.  The verifier converts every $Q_j$ exactly to Bernstein form. Nonnegative Bernstein coefficients imply nonnegativity on the whole interval; when this test is inconclusive, the interval is subdivided at a rational midpoint and the test is repeated.  Since all input coefficients are rational, this part of the verification uses exact rational arithmetic.

At an interface $\xi$ where $\mathfrak r'$ is discontinuous, the verifier also checks
\begin{equation}
    [\mathfrak r']_\xi
    :=
    \mathfrak r'(\xi+)-\mathfrak r'(\xi-)
    \ge0.
\label{eq:upper-jump-check}
\end{equation}
This is the sign required by the generalized It\^o--Tanaka formula in the proof of \cref{prop:zero-upper}.

For the left tail, put $x=-z$ and then $y=x^{-1/2}$, so $0\le y\le1/2$.  After substitution into the residual inequalities and multiplication by positive monomials, each required sign reduces to the sign of a rational polynomial in $y$.  Recursive Bernstein subdivision verifies these signs on the entire interval $[0,1/2]$, and therefore for all $z\le-4$.

For the right tail, write again
\[
    w=Ae^{-(z^2-16)/2}.
\]
The profile inequalities reduce to rational estimates in $z$ and $w$.  The only residual requiring a nonlocal tail estimate has the exact form
\begin{equation}
    \cL\mathfrak r(z)+\lambda\mathfrak r(z)
    =
    C(\lambda-1)zw
    +
    \sum_{k=2}^{6}p_k(z)w^k,
\label{eq:right-tail-remainder-decomposition}
\end{equation}
where every $p_k$ is a rational polynomial of degree at most six. For every monomial $z^mw^k$ appearing in the remainder, division by the positive leading factor $zw$ gives
\[
    \frac{|z^mw^k|}{zw}
    \le
    \begin{cases}
        4^{m-1}A^{k-1},&m\ge1,\\[1mm]
        A^{k-1}/4,&m=0,
    \end{cases}
    \qquad z\ge4.
\]
Indeed, each function $z^{m-1}w^{k-1}$ that occurs here is decreasing on $[4,\infty)$.  The verifier sums these rational upper bounds and checks that the resulting bound on the remainder is strictly smaller than $C(\lambda-1)$.  Hence \eqref{eq:right-tail-remainder-decomposition} is positive on the whole half-line $[4,\infty)$, rather than only at sampled points.

The tail estimates also prove that $g'$ is bounded and that $g$ has at most linear growth.  As a consistency check, the leading left-tail behavior $g(z)\sim a|z|$ together with
\[
    -\Lambda g\le\cL g\le-g
\]
forces
\[
    \frac32\le a\le\Lambda+\frac12,
\]
and the rational value of $a$ in the certificate lies in this interval.

To connect the interface checks with the stochastic argument, recall that for a continuous piecewise $C^2$ function $f$ and a continuous semimartingale $Z$,
\begin{align}
    f(Z_t)
    &=
    f(Z_0)
    +
    \int_0^t f'(Z_s)\,dZ_s
    +
    \frac12
    \int_0^t f''(Z_s)\,d\langle Z\rangle_s
    \notag\\
    &\quad
    +
    \frac12
    \sum_\xi
    [f']_\xi L_t^\xi(Z).
\label{eq:piecewise-ito-tanaka}
\end{align}
Applied to $\mathfrak r$, condition \eqref{eq:upper-jump-check} makes every local-time contribution nonnegative.  The localization used in \cref{prop:zero-upper} may therefore be removed using the growth bounds verified above.

\paragraph{The integral inequality.}
It remains to verify parts~(iv)--(v) of \cref{lem:certified-profile}.  Set
\[
    K:=\frac{R_0}{\lambda-2},
    \qquad
    A_0:=\mathcal E_+-CG_0,
    \qquad
    A_1:=-\left(G_0^2-CG_0-K\right).
\]
For $0\le s\le1$, the coefficient signs in \eqref{eq:vbar}, together with the exact rational inequalities
\[
    2\Lambda-1\le\frac{25}{24},
    \qquad
    \lambda-1\ge\frac32,
\]
give
\begin{equation}
    \mathcal E_+-\overline v(s)
    \ge
    A_0+A_1s+\gamma^2s^{25/24}-Ks^{3/2}.
\label{eq:barrier-rational-envelope}
\end{equation}
With the substitution $s=x^{24}$, the right-hand side becomes the rational polynomial
\[
    P(x)
    =
    A_0+A_1x^{24}+\gamma^2x^{25}-Kx^{36},
    \qquad 0\le x\le1.
\]
The certificate contains a rational partition of $[0,1]$.  On every cell, the verifier proves a positive rational lower bound for $P$ and a rational lower bound for its square root.  Summing the corresponding lower bounds for the transformed integral proves both
\[
    \mathcal E_+-\overline v(s)>0
    \qquad (0\le s\le1)
\]
and
\[
    \int_0^1
    \sqrt{\mathcal E_+-\overline v(s)}\,ds
    >
    \frac{2G_0}{3},
\]
without relying on floating-point quadrature.

\begin{proposition}
\label{prop:upper-certificate-soundness}
If the upper-bound verifier accepts the supplied certificate, then every assertion of \cref{lem:certified-profile} holds.  Consequently,
\[
    \Estar(0)\le\frac{243}{125}=1.944.
\]
\end{proposition}

\begin{proof}
The exact spline, tail, matching, and interface checks prove parts~(i)--(iii) of \cref{lem:certified-profile}.  The polynomial lower bounds and rational lower sum above prove parts~(iv)--(v).  Thus the whole lemma follows from acceptance of the certificate.  The final inequality is then \cref{prop:zero-upper}.
\end{proof}

\subsection{Verification of the zero-margin lower bound}
\label{app:lower-certificate}

The lower-bound certificate consists of the $68$ rational interval lengths and strictly increasing rational weights from \cref{lem:weighted-certificate}, together with the numerical data used to propagate rigorous upper bounds through the recursion \eqref{eq:hjb-terminal}--\eqref{eq:hjb-recursion}.  The optimization that found these weights is not used by the proof.

The verifier works on a uniform rational grid
\[
    x_j=-R+jh
    \qquad (-R\le x_j\le R).
\]
At every stage it stores upper bounds for the exact values $H_i(x_j)$. The main issue is to extend nodewise bounds to the continuum before applying the next heat convolution.

For $d>0$, write
\[
    \ph_d(x)
    :=
    \frac{1}{\sqrt{2\pi d}}
    e^{-x^2/(2d)}
\]
for the density of $N(0,d)$.

\begin{lemma}
\label{lem:heat-curvature}
If $F:\R\to[0,1]$ is nondecreasing and $H=P_dF$, then
\[
    \norm{H''}_\infty
    \le
    \frac{1}{\sqrt{2\pi e}\,d}.
\]
\end{lemma}

\begin{proof}
The distributional derivative $\mu=dF$ is a positive measure of total mass at most one.  Since
\[
    H''=\ph_d'\ast\mu,
\]
we have
\[
    |H''(x)|
    \le
    \norm{\ph_d'}_\infty
    =
    \frac{1}{\sqrt{2\pi e}\,d}.
\]
\end{proof}

Suppose that $U_j\ge H(x_j)$ and $\norm{H''}_\infty\le M$.  For
\[
    x=(1-\theta)x_j+\theta x_{j+1},
    \qquad 0\le\theta\le1,
\]
the linear-interpolation error gives
\begin{equation}
    H(x)
    \le
    (1-\theta)U_j+\theta U_{j+1}
    +
    \frac{Mh^2}{8}.
\label{eq:lower-interpolation-envelope}
\end{equation}
Put
\[
    c:=\frac{Mh^2}{8},
    \qquad
    \widetilde U_j:=U_j+c,
    \qquad
    \widetilde U_{j+1}:=U_{j+1}+c.
\]
Then
\[
    H(x)
    \le
    (1-\theta)\widetilde U_j
    +
    \theta\widetilde U_{j+1}.
\]

For the supplied certificate, every exponent
\[
    r_i:=\frac{a_{i+1}}{a_i}
\]
is greater than one.  Hence $x\mapsto x^{r_i}$ is increasing and convex on $[0,1]$.  Whenever both corrected endpoints are at most one,
\begin{align}
    H(x)^{r_i}
    &\le
    \left(
        (1-\theta)\widetilde U_j
        +
        \theta\widetilde U_{j+1}
    \right)^{r_i}
    \notag\\
    &\le
    (1-\theta)\widetilde U_j^{r_i}
    +
    \theta\widetilde U_{j+1}^{r_i}.
\label{eq:lower-power-envelope}
\end{align}
If a corrected endpoint exceeds one, the verifier uses the trivial upper bound $1$ instead.  On the two exterior half-lines it uses monotonicity on the left and $0\le H_i\le1$ on the right.  Thus the nodewise data determine a global piecewise-affine upper bound for $H_{i+1}^{r_i}$.

\paragraph{Exact convolution of a piecewise-affine upper bound.}
Consider a grid cell $[x_{k-1},x_k]$ whose affine upper bound takes values $L_k$ and $R_k$ at the two endpoints.  When evaluating its heat convolution at $x_j$, put $q=k-j$ and define
\begin{align*}
    Q_q(d)
    &:=
    \Ph\!\left(\frac{qh}{\sqrt d}\right)
    -
    \Ph\!\left(\frac{(q-1)h}{\sqrt d}\right),
\\
    M_q(d)
    &:=
    \frac d h
    \left(
        \ph_d((q-1)h)
        -
        \ph_d(qh)
    \right),
\\
    A_q(d)
    &:=
    qQ_q(d)-M_q(d),
\\
    B_q(d)
    &:=
    M_q(d)-(q-1)Q_q(d).
\end{align*}
Direct integration of the two affine basis functions gives the exact cell contribution
\begin{equation}
    L_kA_q(d)+R_kB_q(d).
\label{eq:exact-affine-cell-convolution}
\end{equation}
Both coefficients are nonnegative because they are integrals of nonnegative barycentric coordinates against the Gaussian kernel.  The two exterior constant pieces contribute Gaussian tail probabilities.  It follows that upward rounding of every endpoint value, convolution coefficient, and exterior-tail contribution preserves the upper-bound direction.

The verifier stores all nodewise upper bounds on a dyadic lattice. Gaussian tails, Gaussian densities, powers, and the final logarithm are evaluated with directed MPFR rounding.  The nonnegative integer convolution corresponding to \eqref{eq:exact-affine-cell-convolution} is computed exactly by zero-padded number-theoretic transforms and Chinese remaindering. At startup, the verifier checks that the transform lengths are supported and that the product of the CRT moduli exceeds a worst-case coefficient bound. These implementation choices affect running time but not the induction establishing soundness.

\begin{proposition}
\label{prop:lower-certificate-soundness}
If the lower-bound verifier accepts the supplied certificate, then at every stage $i$ and every grid point $x_j$, its stored value is an upper bound for the exact value $H_i(x_j)$.  Consequently,
\[
    -2a_0\log H_0(0)
    >
    \frac{173639}{100000},
\]
and therefore
\[
    \Estar(0)
    >
    \frac{173639}{100000}.
\]
\end{proposition}

\begin{proof}
At the terminal stage, directed evaluation of the Gaussian distribution function gives valid upper bounds at all grid points.

Assume inductively that the stored values at stage $i+1$ are upper bounds for $H_{i+1}$ at the grid points. \Cref{lem:heat-curvature,eq:lower-interpolation-envelope} extend these nodewise bounds to an upper bound between adjacent nodes. Convexity of the power map, together with the exterior bounds, then gives a global piecewise-affine upper bound for
\[
    H_{i+1}^{a_{i+1}/a_i}.
\]
Formula \eqref{eq:exact-affine-cell-convolution}, the nonnegativity of its coefficients, and directed rounding produce valid nodewise upper bounds for $H_i=P_{\Delta_i}(H_{i+1}^{a_{i+1}/a_i})$. Backward induction therefore gives
\[
    H_0(0)\le U_0,
\]
where $U_0$ is the final value accepted by the verifier.

Using the exact rational value of $a_0$ and a directed lower bound for $-\log U_0$, the verifier establishes
\[
    -2a_0\log U_0
    >
    \frac{173639}{100000}.
\]
Since $-\log$ is decreasing,
\[
    -2a_0\log H_0(0)
    \ge
    -2a_0\log U_0
    >
    \frac{173639}{100000}.
\]
Combining this with \cref{lem:weighted-lower-bound,thm:weighted-recursion} proves the claimed lower bound for $\Estar(0)$.
\end{proof}

\subsection{Artifact and reproducibility}
\label{app:artifact}

The artifact is available at \url{https://github.com/ainta/exact-online-threshold-artifact}. A single script runs both verifiers; acceptance, combined with \cref{prop:upper-certificate-soundness,prop:lower-certificate-soundness}, establishes the two bounds in \eqref{eq:zero-E}.

\end{document}